\documentclass[acmsmall]{acmart}

\makeatletter
\@ifundefined{lhs2tex.lhs2tex.sty.read}%
  {\@namedef{lhs2tex.lhs2tex.sty.read}{}%
   \newcommand\SkipToFmtEnd{}%
   \newcommand\EndFmtInput{}%
   \long\def\SkipToFmtEnd#1\EndFmtInput{}%
  }\SkipToFmtEnd

\newcommand\ReadOnlyOnce[1]{\@ifundefined{#1}{\@namedef{#1}{}}\SkipToFmtEnd}
\usepackage{amstext}
\usepackage{amssymb}
\usepackage{stmaryrd}
\DeclareFontFamily{OT1}{cmtex}{}
\DeclareFontShape{OT1}{cmtex}{m}{n}
  {<5><6><7><8>cmtex8
   <9>cmtex9
   <10><10.95><12><14.4><17.28><20.74><24.88>cmtex10}{}
\DeclareFontShape{OT1}{cmtex}{m}{it}
  {<-> ssub * cmtt/m/it}{}

\DeclareFontShape{OT1}{cmtt}{bx}{n}
  {<5><6><7><8>cmtt8
   <9>cmbtt9
   <10><10.95><12><14.4><17.28><20.74><24.88>cmbtt10}{}
\DeclareFontShape{OT1}{cmtex}{bx}{n}
  {<-> ssub * cmtt/bx/n}{}

\newcommand{\Conid}[1]{\mathit{#1}}
\newcommand{\Varid}[1]{\mathit{#1}}
\newcommand{\anonymous}{\kern0.06em \vbox{\hrule\@width.5em}}

\usepackage{polytable}

\@ifundefined{mathindent}%
  {\newdimen\mathindent\mathindent\leftmargini}%
  {}%

\def\resethooks{%
  \global\let\SaveRestoreHook\empty
  \global\let\ColumnHook\empty}
\newcommand*{\savecolumns}[1][default]%
  {\g@addto@macro\SaveRestoreHook{\savecolumns[#1]}}
\newcommand*{\restorecolumns}[1][default]%
  {\g@addto@macro\SaveRestoreHook{\restorecolumns[#1]}}
\newcommand*{\aligncolumn}[2]%
  {\g@addto@macro\ColumnHook{\column{#1}{#2}}}

\resethooks

\newcommand{\onelinecommentchars}{\quad-{}- }
\newcommand{\commentbeginchars}{\enskip\{-}
\newcommand{\commentendchars}{-\}\enskip}

\newcommand{\visiblecomments}{%
  \let\onelinecomment=\onelinecommentchars
  \let\commentbegin=\commentbeginchars
  \let\commentend=\commentendchars}

\newcommand{\invisiblecomments}{%
  \let\onelinecomment=\empty
  \let\commentbegin=\empty
  \let\commentend=\empty}

\visiblecomments

\newlength{\blanklineskip}
\newcommand{\hsindent}[1]{\quad}
\let\hspre\empty
\let\hspost\empty

\EndFmtInput
\makeatother
\ReadOnlyOnce{forall.fmt}%
\makeatletter

\let\HaskellResetHook\empty
\newcommand*{\AtHaskellReset}[1]{%
  \g@addto@macro\HaskellResetHook{#1}}
\newcommand*{\HaskellReset}{\HaskellResetHook}

\newcommand\hsforall{\global\let\hsdot=\hsperiodonce}
\newcommand*\hsperiodonce[2]{#2\global\let\hsdot=\hscompose}
\newcommand*\hscompose[2]{#1}

\AtHaskellReset{\global\let\hsdot=\hscompose}

\HaskellReset

\makeatother
\EndFmtInput
\ReadOnlyOnce{polycode.fmt}%
\makeatletter

\newcommand{\hsnewpar}[1]%
  {{\parskip=0pt\parindent=0pt\par\vskip #1\noindent}}

\newcommand{\hscodestyle}{}

\newcommand{\sethscode}[1]%
  {\expandafter\let\expandafter\hscode\csname #1\endcsname
   \expandafter\let\expandafter\endhscode\csname end#1\endcsname}

  {\par\noindent
   \advance\leftskip\mathindent
   \hscodestyle
   \let\\=\@normalcr
   \let\hspre\(\let\hspost\)%
   \pboxed}%
  {\endpboxed\)%
   \par\noindent
   \ignorespacesafterend}

  {\hsnewpar\abovedisplayskip
   \advance\leftskip\mathindent
   \hscodestyle
   \let\hspre\(\let\hspost\)%
   \pboxed}%
  {\endpboxed%
   \hsnewpar\belowdisplayskip
   \ignorespacesafterend}

  {\hsnewpar\abovedisplayskip
   \advance\leftskip\mathindent
   \hscodestyle
   \let\\=\@normalcr
   \(\pboxed}%
  {\endpboxed\)%
   \hsnewpar\belowdisplayskip
   \ignorespacesafterend}

\newcommand{\plainhs}{\sethscode{plainhscode}}

\plainhs

  {\hsnewpar\abovedisplayskip
   \advance\leftskip\mathindent
   \hscodestyle
   \let\\=\@normalcr
   \(\parray}%
  {\endparray\)%
   \hsnewpar\belowdisplayskip
   \ignorespacesafterend}

  {\parray}{\endparray}

  {\(\parray}{\endparray\)}

\def\codeframewidth{\arrayrulewidth}
\RequirePackage{calc}

  {\parskip=\abovedisplayskip\par\noindent
   \hscodestyle
   \arrayrulewidth=\codeframewidth
   \tabular{@{}|p{\linewidth-2\arraycolsep-2\arrayrulewidth-2pt}|@{}}%
   \hline\framedhslinecorrect\\{-1.5ex}%
   \let\endoflinesave=\\
   \let\\=\@normalcr
   \(\pboxed}%
  {\endpboxed\)%
   \framedhslinecorrect\endoflinesave{.5ex}\hline
   \endtabular
   \parskip=\belowdisplayskip\par\noindent
   \ignorespacesafterend}

\newcommand{\framedhslinecorrect}[2]%
  {#1[#2]}

  {\(\def\column##1##2{}%
   \let\>\undefined\let\<\undefined\let\\\undefined
   \newcommand\>[1][]{}\newcommand\<[1][]{}\newcommand\\[1][]{}%
   \def\fromto##1##2##3{##3}%
   }{\) }%

  {\let\orighscode=\hscode
   \let\origendhscode=\endhscode
   \def\endhscode{\def\hscode{\endgroup\def\@currenvir{hscode}\\}\begingroup}
   \orighscode\def\hscode{\endgroup\def\@currenvir{hscode}}}%
  {\origendhscode
   \global\let\hscode=\orighscode
   \global\let\endhscode=\origendhscode}%

\makeatother
\EndFmtInput

\ReadOnlyOnce{Formatting.fmt}%
\makeatletter

\let\Varid\mathit
\let\Conid\mathsf

\def\commentbegin{\quad\begingroup\color{teal}\{\ }
\def\commentend{\}\endgroup}

\makeatother
\EndFmtInput

\usepackage{amsmath}
\usepackage{mathpartir}
\usepackage{amsfonts}
\usepackage{stmaryrd}
\usepackage{hyperref}
\usepackage{scalerel}
\usepackage{bussproofs}
\EnableBpAbbreviations
\usepackage{url}
\usepackage{subfig}
\usepackage{enumitem}
\usepackage{mdframed}
\usepackage{multicol}
\usepackage{graphicx}

\usepackage{doubleequals}

\newcommand{\delete}[1]{}

\makeatother

\setcopyright{none}             

\allowdisplaybreaks

\definecolor{airforceblue}{rgb}{0.36, 0.54, 0.66}
\definecolor{bleudefrance}{rgb}{0.19, 0.55, 0.91}
\definecolor{blue(ncs)}{rgb}{0.0, 0.53, 0.74}
\definecolor{mediumpersianblue}{rgb}{0.0, 0.4, 0.65}
\everymath{\color{mediumpersianblue}}

\begin{document}

\title[Calculating a Backtracking Algorithm]%
{Calculating a Backtracking Algorithm:}
\subtitle{An Exercise in Monadic Program Derivation}


\author{Shin-Cheng Mu}
\affiliation{
  \department{Institute of Information Science}              
  \institution{Academia Sinica}            
  \country{Taiwan}
}
\email{scm@iis.sinica.edu.tw}          

\begin{abstract}
Equational reasoning is among the most important tools that functional programming provides us.
Curiously, relatively less attention has been paid to reasoning about monadic programs.
In this report we derive a backtracking algorithm for problem specifications
that use a monadic unfold to generate possible solutions, which are filtered using a \ensuremath{\Varid{scanl}}-like predicate.
We develop theorems that convert a variation of \ensuremath{\Varid{scanl}} to a \ensuremath{\Varid{foldr}} that uses the state monad, as well as theorems constructing hylomorphism.
The algorithm is used to solve the \ensuremath{\Varid{n}}-queens puzzle, our running example.
The aim is to develop theorems and patterns useful for the derivation of monadic programs, focusing on the intricate interaction between state and non-determinism.
\end{abstract}

\maketitle

\section{Introduction}

Equational reasoning is among the many gifts that functional programming offers us. Functional programs preserve a rich set of mathematical properties, which not only helps to prove properties about programs in a relatively simple and elegant manner, but also aids the development of programs. One may refine a clear but inefficient specification, stepwise through equational reasoning, to an efficient program whose correctness may not be obvious without such a derivation.

It is misleading if one says that functional programming does not allow side effects.
In fact, even a purely functional language may allow a variety of side effects --- in a rigorous, mathematically manageable manner.
Since the introduction of {\em monads} into the functional programming community~\cite{Moggi:89:Computational, Wadler:92:Monads}, it has become the main framework in which effects are modelled.
Various monads were developed for different effects, from general ones such as IO, state, non-determinism, exception, continuation, environment passing, to specific purposes such as parsing.
Numerous research were also devoted to producing practical monadic programs.

It is also a wrong impression that impure programs are bound to be difficult to reason about.
In fact, the laws of monads and their operators are sufficient to prove quite a number of useful properties about monadic programs.
The validity of these properties, proved using only these laws, is independent from the particular implementation of the monad.

This report follows the trail of Hutton and Fulger~\shortcite{HuttonFulger:08:Reasoning} and~Gibbons and Hinze~\shortcite{GibbonsHinze:11:Just}, aiming to develop theorems and patterns that are useful for reasoning about monadic programs.
We focus on two effects --- non-determinism and state.
In this report we consider problem specifications that use a monadic unfold to generate possible solutions, which are filtered using a \ensuremath{\Varid{scanl}}-like predicate.
We develop theorems that convert a variation of \ensuremath{\Varid{scanl}} to a \ensuremath{\Varid{foldr}} that uses the state monad, as well as theorems constructing hylomorphism.
The algorithm is used to solve the \ensuremath{\Varid{n}}-queens puzzle, our running example.

While the interaction between non-determinism and state is known to be intricate,
when each non-deterministic branch has its own local state, we get a relatively well-behaved monad that provides a rich collection of properties to work with.
The situation when the state is global and shared by all non-deterministic branches is much more complex, and is dealt with in a subsequent paper \citep{Pauwels:19:Handling}.
\section{Monad and Effect Operators}

A monad consists of a type constructor \ensuremath{\Varid{m}\mathbin{::}\mathbin{*}\to \mathbin{*}} and two operators
\ensuremath{\Varid{return}} and ``bind'' \ensuremath{(\mathrel{\hstretch{0.7}{>\!\!>\!\!=}})}, often modelled by the following Haskell type class declaration:%
\footnote{
This report uses type classes to be explicit about the effects a program uses.
For example, programs using non-determinism are labelled with constraint \ensuremath{\Conid{MonadPlus}\;\Varid{m}}.
The style of reasoning proposed in this report is not tied to type classes or Haskell,
and we do not strictly follow the particularities of type classes in the current Haskell standard.
For example, we overlook the particularities that a \ensuremath{\Conid{Monad}} must also be \ensuremath{\Conid{Applicative}}, \ensuremath{\Conid{MonadPlus}} be \ensuremath{\Conid{Alternative}}, and that functional dependency is needed in a number of places in this report.
}
\begin{hscode}\SaveRestoreHook
\column{B}{@{}>{\hspre}l<{\hspost}@{}}%
\column{5}{@{}>{\hspre}l<{\hspost}@{}}%
\column{13}{@{}>{\hspre}l<{\hspost}@{}}%
\column{22}{@{}>{\hspre}l<{\hspost}@{}}%
\column{E}{@{}>{\hspre}l<{\hspost}@{}}%
\>[B]{}\mathbf{class}\;\Conid{Monad}\;\Varid{m}\;\mathbf{where}{}\<[E]%
\\
\>[B]{}\hsindent{5}{}\<[5]%
\>[5]{}\Varid{return}{}\<[13]%
\>[13]{}\mathbin{::}\Varid{a}\to {}\<[22]%
\>[22]{}\Varid{m}\;\Varid{a}{}\<[E]%
\\
\>[B]{}\hsindent{5}{}\<[5]%
\>[5]{}(\mathrel{\hstretch{0.7}{>\!\!>\!\!=}}){}\<[13]%
\>[13]{}\mathbin{::}\Varid{m}\;\Varid{a}\to (\Varid{a}\to \Varid{m}\;\Varid{b})\to \Varid{m}\;\Varid{b}~~.{}\<[E]%
\ColumnHook
\end{hscode}\resethooks
They are supposed to satisfy the following {\em monad laws}:
\begin{align}
  \ensuremath{\Varid{return}\;\Varid{x}\mathrel{\hstretch{0.7}{>\!\!>\!\!=}}\Varid{f}} &= \ensuremath{\Varid{f}\;\Varid{x}}\mbox{~~,} \label{eq:monad-bind-ret}\\
  \ensuremath{\Varid{m}\mathrel{\hstretch{0.7}{>\!\!>\!\!=}}\Varid{return}} &= \ensuremath{\Varid{m}} \mbox{~~,} \label{eq:monad-ret-bind}\\
  \ensuremath{(\Varid{m}\mathrel{\hstretch{0.7}{>\!\!>\!\!=}}\Varid{f})\mathrel{\hstretch{0.7}{>\!\!>\!\!=}}\Varid{g}} &= \ensuremath{\Varid{m}\mathrel{\hstretch{0.7}{>\!\!>\!\!=}}(\lambda \Varid{x}\to \Varid{f}\;\Varid{x}\mathrel{\hstretch{0.7}{>\!\!>\!\!=}}\Varid{g})} \mbox{~~.} \label{eq:monad-assoc}
\end{align}
We also define \ensuremath{\Varid{m}_{1}\mathbin{\hstretch{0.7}{>\!\!>}}\Varid{m}_{2}\mathrel{=}\Varid{m}_{1}\mathrel{\hstretch{0.7}{>\!\!>\!\!=}}\Varid{const}\;\Varid{m}_{2}}, which has type \ensuremath{(\mathbin{\hstretch{0.7}{>\!\!>}})\mathbin{::}\Varid{m}\;\Varid{a}\to \Varid{m}\;\Varid{b}\to \Varid{m}\;\Varid{b}}.
Kleisli composition, denoted by \ensuremath{(\mathrel{\hstretch{0.7}{>\!\!=\!\!\!>}})}, composes two monadic operations \ensuremath{\Varid{a}\to \Varid{m}\;\Varid{b}} and \ensuremath{\Varid{b}\to \Varid{m}\;\Varid{c}} into an operation \ensuremath{\Varid{a}\to \Varid{m}\;\Varid{c}}.
The operator \ensuremath{(\mathrel{\raisebox{0.5\depth}{\scaleobj{0.5}{\langle}} \scaleobj{0.8}{\$} \raisebox{0.5\depth}{\scaleobj{0.5}{\rangle}}})} applies a pure function to a monad.
%
%
\begin{hscode}\SaveRestoreHook
\column{B}{@{}>{\hspre}l<{\hspost}@{}}%
\column{8}{@{}>{\hspre}l<{\hspost}@{}}%
\column{E}{@{}>{\hspre}l<{\hspost}@{}}%
\>[B]{}(\mathrel{\hstretch{0.7}{>\!\!=\!\!\!>}}){}\<[8]%
\>[8]{}\mathbin{::}\Conid{Monad}\;\Varid{m}\Rightarrow (\Varid{a}\to \Varid{m}\;\Varid{b})\to (\Varid{b}\to \Varid{m}\;\Varid{c})\to \Varid{a}\to \Varid{m}\;\Varid{c}{}\<[E]%
\\
\>[B]{}(\Varid{f}\mathrel{\hstretch{0.7}{>\!\!=\!\!\!>}}\Varid{g})\;\Varid{x}\mathrel{=}\Varid{f}\;\Varid{x}\mathrel{\hstretch{0.7}{>\!\!>\!\!=}}\Varid{g}~~,{}\<[E]%
\\[\blanklineskip]%
\>[B]{}(\mathrel{\raisebox{0.5\depth}{\scaleobj{0.5}{\langle}} \scaleobj{0.8}{\$} \raisebox{0.5\depth}{\scaleobj{0.5}{\rangle}}}){}\<[8]%
\>[8]{}\mathbin{::}\Conid{Monad}\;\Varid{m}\Rightarrow (\Varid{a}\to \Varid{b})\to \Varid{m}\;\Varid{a}\to \Varid{m}\;\Varid{b}{}\<[E]%
\\
\>[B]{}\Varid{f}\mathrel{\raisebox{0.5\depth}{\scaleobj{0.5}{\langle}} \scaleobj{0.8}{\$} \raisebox{0.5\depth}{\scaleobj{0.5}{\rangle}}}\Varid{n}\mathrel{=}\Varid{n}\mathrel{\hstretch{0.7}{>\!\!>\!\!=}}(\Varid{return}\mathbin{\cdot}\Varid{f})~~.{}\<[E]%
\ColumnHook
\end{hscode}\resethooks
The following properties can be proved from their definitions and the monad laws:
\begin{align}
  \ensuremath{(\Varid{f}\mathbin{\cdot}\Varid{g})\mathrel{\raisebox{0.5\depth}{\scaleobj{0.5}{\langle}} \scaleobj{0.8}{\$} \raisebox{0.5\depth}{\scaleobj{0.5}{\rangle}}}\Varid{m}} &= \ensuremath{\Varid{f}\mathrel{\raisebox{0.5\depth}{\scaleobj{0.5}{\langle}} \scaleobj{0.8}{\$} \raisebox{0.5\depth}{\scaleobj{0.5}{\rangle}}}(\Varid{g}\mathrel{\raisebox{0.5\depth}{\scaleobj{0.5}{\langle}} \scaleobj{0.8}{\$} \raisebox{0.5\depth}{\scaleobj{0.5}{\rangle}}}\Varid{m})} \mbox{~~,}
    \label{eq:comp-ap-ap}\\
\ensuremath{(\Varid{f}\mathrel{\raisebox{0.5\depth}{\scaleobj{0.5}{\langle}} \scaleobj{0.8}{\$} \raisebox{0.5\depth}{\scaleobj{0.5}{\rangle}}}\Varid{m})\mathrel{\hstretch{0.7}{>\!\!>\!\!=}}\Varid{g}} &= \ensuremath{\Varid{m}\mathrel{\hstretch{0.7}{>\!\!>\!\!=}}(\Varid{g}\mathbin{\cdot}\Varid{f})} \mbox{~~,}
  \label{eq:comp-bind-ap}\\
\ensuremath{\Varid{f}\mathrel{\raisebox{0.5\depth}{\scaleobj{0.5}{\langle}} \scaleobj{0.8}{\$} \raisebox{0.5\depth}{\scaleobj{0.5}{\rangle}}}(\Varid{m}\mathrel{\hstretch{0.7}{>\!\!>\!\!=}}\Varid{k})} &= \ensuremath{\Varid{m}\mathrel{\hstretch{0.7}{>\!\!>\!\!=}}(\lambda \Varid{x}\to \Varid{f}\mathrel{\raisebox{0.5\depth}{\scaleobj{0.5}{\langle}} \scaleobj{0.8}{\$} \raisebox{0.5\depth}{\scaleobj{0.5}{\rangle}}}\Varid{k}\;\Varid{x})}  \mbox{~~, \ensuremath{\Varid{x}} not free in \ensuremath{\Varid{f}}.}
  \label{eq:ap-bind-ap}
\end{align}

\paragraph{Effect and Effect Operators}
Monads are used to model effects, and each effect comes with its collection of operators. For example, to model non-determinism we assume two operators \ensuremath{\emptyset} and \ensuremath{(\talloblong)}, respectively modeling failure and choice.
A state effect comes with operators \ensuremath{\Varid{get}} and \ensuremath{\Varid{put}}, which respectively reads from and writes to an unnamed state variable.

A program may involve more than one effect.
In Haskell, the type class constraint \ensuremath{\Conid{MonadPlus}} in the type of a program denotes that the program may use \ensuremath{\emptyset} or \ensuremath{(\talloblong)}, and possibly other effects, while \ensuremath{\Conid{MonadState}\;\Varid{s}} denotes that it may use \ensuremath{\Varid{get}} and \ensuremath{\Varid{put}}.
Some theorems in this report, however, apply only to programs that, for example, use non-determinism and \emph{no other effects}.
In such cases we will note in text that the theorem applies only to programs ``whose only effect is non-determinism.''
The set of effects a program uses can always be statically inferred by syntax.

\paragraph{Total, Finite Programs} Like in other literature on program derivation, we assume a set-theoretic semantics in which functions are total. We thus have the following laws regarding branching:
\begin{align}
  \ensuremath{\Varid{f}\;(\mathbf{if}\;\Varid{p}\;\mathbf{then}\;\Varid{e}_{1}\;\mathbf{else}\;\Varid{e}_{2})} &= \ensuremath{\mathbf{if}\;\Varid{p}\;\mathbf{then}\;\Varid{f}\;\Varid{e}_{1}\;\mathbf{else}\;\Varid{f}\;\Varid{e}_{2}} \mbox{~~,}
  \label{eq:if-distr}\\
  \ensuremath{\mathbf{if}\;\Varid{p}\;\mathbf{then}\;(\lambda \Varid{x}\to \Varid{e}_{1})\;\mathbf{else}\;(\lambda \Varid{x}\to \Varid{e}_{2})} &=
    \ensuremath{\lambda \Varid{x}\to \mathbf{if}\;\Varid{p}\;\mathbf{then}\;\Varid{e}_{1}\;\mathbf{else}\;\Varid{e}_{2}} \mbox{~~.}\label{eq:if-fun}
\end{align}
Lists in this report are inductive types, and unfolds generate finite lists too. Non-deterministic choices are finitely branching.
Given a concrete input, a function always expands to a finitely-sized expression consisting of syntax allowed by its type. We may therefore prove properties of a monadic program by structural induction over its syntax.

\section{Example: The \ensuremath{\Varid{n}}-Queens Problem}
\label{sec:queens}

Reasoning about monadic programs gets more interesting when more than one effect is involved.
Backtracking algorithms make good examples of programs that are stateful and non-deterministic, and the \ensuremath{\Varid{n}}-queens problem, also dealt with by Gibbons and Hinze~\shortcite{GibbonsHinze:11:Just}, is among the most well-known examples of backtracking.\footnote{Curiously, Gibbons and Hinze~\shortcite{GibbonsHinze:11:Just} did not finish their derivation and stopped at a program that exhaustively generates all permutations and tests each of them. Perhaps it was sufficient to demonstrate their point.}

In this section we present a specification of the problem, before transforming it into the form \ensuremath{\Varid{unfoldM}\;\Varid{p}\;\Varid{f}\mathrel{\hstretch{0.7}{>\!\!=\!\!\!>}}\Varid{filt}\;(\Varid{all}\;\Varid{ok}\mathbin{\cdot}\Varid{scanl}_{+}\;(\oplus)\;\Varid{st})} (whose components will be defined later), which is the general form of problems we will deal with in this report.

\subsection{Non-Determinism}

Since the \ensuremath{\Varid{n}}-queens problem will be specified by a non-deterministic program,
we discuss non-determinism before presenting the specification.
We assume two operators \ensuremath{\emptyset} and \ensuremath{(\talloblong)}:
\begin{hscode}\SaveRestoreHook
\column{B}{@{}>{\hspre}l<{\hspost}@{}}%
\column{3}{@{}>{\hspre}l<{\hspost}@{}}%
\column{10}{@{}>{\hspre}l<{\hspost}@{}}%
\column{E}{@{}>{\hspre}l<{\hspost}@{}}%
\>[B]{}\mathbf{class}\;\Conid{Monad}\;\Varid{m}\Rightarrow \Conid{MonadPlus}\;\Varid{m}\;\mathbf{where}{}\<[E]%
\\
\>[B]{}\hsindent{3}{}\<[3]%
\>[3]{}\emptyset{}\<[10]%
\>[10]{}\mathbin{::}\Varid{m}\;\Varid{a}{}\<[E]%
\\
\>[B]{}\hsindent{3}{}\<[3]%
\>[3]{}(\talloblong){}\<[10]%
\>[10]{}\mathbin{::}\Varid{m}\;\Varid{a}\to \Varid{m}\;\Varid{a}\to \Varid{m}\;\Varid{a}~~.{}\<[E]%
\ColumnHook
\end{hscode}\resethooks
The former denotes failure, while \ensuremath{\Varid{m}\mathbin{\talloblong}\Varid{n}} denotes that the computation may yield either \ensuremath{\Varid{m}} or \ensuremath{\Varid{n}}. What laws they should satisfy, however, can be a tricky issue. As discussed by Kiselyov~\shortcite{Kiselyov:15:Laws}, it eventually comes down to what we use the monad for. It is usually expected that \ensuremath{(\talloblong)} and \ensuremath{\emptyset} form a monoid. That is, \ensuremath{(\talloblong)} is associative, with \ensuremath{\emptyset} as its zero:
\begin{align}
\ensuremath{(\Varid{m}\mathbin{\talloblong}\Varid{n})\mathbin{\talloblong}\Varid{k}}~ &=~ \ensuremath{\Varid{m}\mathbin{\talloblong}(\Varid{n}\mathbin{\talloblong}\Varid{k})} \mbox{~~,}
  \label{eq:mplus-assoc}\\
\ensuremath{\emptyset\mathbin{\talloblong}\Varid{m}} ~=~ & \ensuremath{\Varid{m}} ~=~ \ensuremath{\Varid{m}\mathbin{\talloblong}\emptyset} \mbox{~~.}
  \label{eq:mzero-mplus}
\end{align}
It is also assumed that monadic bind distributes into \ensuremath{(\talloblong)} from the end,
while \ensuremath{\emptyset} is a left zero for \ensuremath{(\mathrel{\hstretch{0.7}{>\!\!>\!\!=}})}:
\begin{alignat}{2}
  &\mbox{\bf left-distributivity}:\quad &
  \ensuremath{(\Varid{m}_{1}\mathbin{\talloblong}\Varid{m}_{2})\mathrel{\hstretch{0.7}{>\!\!>\!\!=}}\Varid{f}} ~&=~ \ensuremath{(\Varid{m}_{1}\mathrel{\hstretch{0.7}{>\!\!>\!\!=}}\Varid{f})\mathbin{\talloblong}(\Varid{m}_{2}\mathrel{\hstretch{0.7}{>\!\!>\!\!=}}\Varid{f})} \mbox{~~,}
  \label{eq:bind-mplus-dist}\\
  &\mbox{\bf left-zero}:\quad &
  \ensuremath{\emptyset\mathrel{\hstretch{0.7}{>\!\!>\!\!=}}\Varid{f}} ~&=~ \ensuremath{\emptyset} \label{eq:bind-mzero-zero} \mbox{~~.}
\end{alignat}
We will refer to the laws \eqref{eq:mplus-assoc}, \eqref{eq:mzero-mplus},
\eqref{eq:bind-mplus-dist}, \eqref{eq:bind-mzero-zero} collectively as the
\emph{nondeterminism laws}.
Other properties regarding \ensuremath{\emptyset} and \ensuremath{(\talloblong)} will be introduced when needed.

The monadic function \ensuremath{\Varid{filt}\;\Varid{p}\;\Varid{x}} returns \ensuremath{\Varid{x}} if \ensuremath{\Varid{p}\;\Varid{x}} holds, and fails otherwise:
\begin{hscode}\SaveRestoreHook
\column{B}{@{}>{\hspre}l<{\hspost}@{}}%
\column{E}{@{}>{\hspre}l<{\hspost}@{}}%
\>[B]{}\Varid{filt}\mathbin{::}\Conid{MonadPlus}\;\Varid{m}\Rightarrow (\Varid{a}\to \Conid{Bool})\to \Varid{a}\to \Varid{m}\;\Varid{a}{}\<[E]%
\\
\>[B]{}\Varid{filt}\;\Varid{p}\;\Varid{x}\mathrel{=}\Varid{guard}\;(\Varid{p}\;\Varid{x})\mathbin{\hstretch{0.7}{>\!\!>}}\Varid{return}\;\Varid{x}~~,{}\<[E]%
\ColumnHook
\end{hscode}\resethooks
where \ensuremath{\Varid{guard}} is a standard monadic function defined by:
\begin{hscode}\SaveRestoreHook
\column{B}{@{}>{\hspre}l<{\hspost}@{}}%
\column{E}{@{}>{\hspre}l<{\hspost}@{}}%
\>[B]{}\Varid{guard}\mathbin{::}\Conid{MonadPlus}\;\Varid{m}\Rightarrow \Conid{Bool}\to \Varid{m}\;(){}\<[E]%
\\
\>[B]{}\Varid{guard}\;\Varid{b}\mathrel{=}\mathbf{if}\;\Varid{b}\;\mathbf{then}\;\Varid{return}\;()\;\mathbf{else}\;\emptyset~~.{}\<[E]%
\ColumnHook
\end{hscode}\resethooks
The following properties allow us to move \ensuremath{\Varid{guard}} around.
Their proofs are given in Appendix~\ref{sec:misc-proofs}.
\begin{align}
  \ensuremath{\Varid{guard}\;(\Varid{p}\mathrel{\wedge}\Varid{q})} ~&=~ \ensuremath{\Varid{guard}\;\Varid{p}\mathbin{\hstretch{0.7}{>\!\!>}}\Varid{guard}\;\Varid{q}} \mbox{~~,}
  \label{eq:guard-conj} \\
  \ensuremath{\Varid{guard}\;\Varid{p}\mathbin{\hstretch{0.7}{>\!\!>}}(\Varid{f}\mathrel{\raisebox{0.5\depth}{\scaleobj{0.5}{\langle}} \scaleobj{0.8}{\$} \raisebox{0.5\depth}{\scaleobj{0.5}{\rangle}}}\Varid{m})} ~&=~ \ensuremath{\Varid{f}\mathrel{\raisebox{0.5\depth}{\scaleobj{0.5}{\langle}} \scaleobj{0.8}{\$} \raisebox{0.5\depth}{\scaleobj{0.5}{\rangle}}}(\Varid{guard}\;\Varid{p}\mathbin{\hstretch{0.7}{>\!\!>}}\Varid{m})} \mbox{~~.}
  \label{eq:guard-fmap} \\
  \ensuremath{\Varid{guard}\;\Varid{p}\mathbin{\hstretch{0.7}{>\!\!>}}\Varid{m}}~=~ \ensuremath{\Varid{m}\mathrel{\hstretch{0.7}{>\!\!>\!\!=}}}& \,\ensuremath{(\lambda \Varid{x}\to \Varid{guard}\;\Varid{p}\mathbin{\hstretch{0.7}{>\!\!>}}\Varid{return}\;\Varid{x})}
  \mbox{,~~~~if~} \ensuremath{\Varid{m}\mathbin{\hstretch{0.7}{>\!\!>}}\emptyset\mathrel{=}\emptyset} \mbox{~~.}
     \label{eq:guard-commute}
\end{align}

\subsection{Specification}
\label{sec:queens-spec}

{\arraycolsep=1.4pt
\begin{figure}
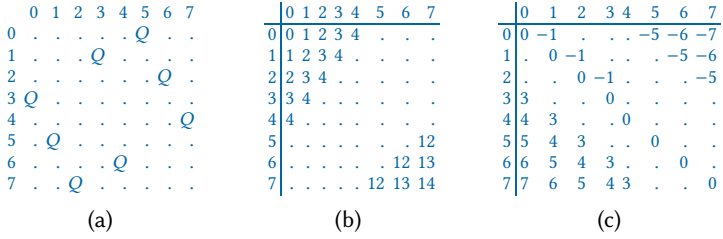

\centering
\subfloat[]{
$\scriptsize
\begin{array}{rrrrrrrrr}
  & 0 & 1 & 2 & 3 & 4 & 5 & 6 & 7\\
0 & . & . & . & . & . & Q & . & .\\
1 & . & . & . & Q & . & . & . & .\\
2 & . & . & . & . & . & . & Q & .\\
3 & Q & . & . & . & . & . & . & .\\
4 & . & . & . & . & . & . & . & Q\\
5 & . & Q & . & . & . & . & . & .\\
6 & . & . & . & . & Q & . & . & .\\
7 & . & . & Q & . & . & . & . & .
\end{array}$
} 
\qquad
\subfloat[]{
$\scriptsize
\begin{array}{r|rrrrrrrr}
  & 0 & 1 & 2 & 3 & 4 & 5 & 6 & 7\\ \hline
0 & 0 & 1 & 2 & 3 & 4 & . & . & .\\
1 & 1 & 2 & 3 & 4 & . & . & . & .\\
2 & 2 & 3 & 4 & . & . & . & . & .\\
3 & 3 & 4 & . & . & . & . & . & .\\
4 & 4 & . & . & . & . & . & . & .\\
5 & . & . & . & . & . & . & . & 12\\
6 & . & . & . & . & . & . & 12& 13\\
7 & . & . & . & . & . & 12& 13& 14
\end{array}
$} 
\qquad
\subfloat[]{
$\scriptsize
\begin{array}{r|rrrrrrrr}
  & 0 & 1 & 2 & 3 & 4 & 5 & 6 & 7\\ \hline
0 & 0 &-1 & . & . & . &-5 &-6 &-7\\
1 & . & 0 &-1 & . & . & . &-5 &-6\\
2 & . & . & 0 &-1 & . & . & . &-5\\
3 & 3 & . & . & 0 & . & . & . & .\\
4 & 4 & 3 & . & . & 0 & . & . & .\\
5 & 5 & 4 & 3 & . & . & 0 & . & .\\
6 & 6 & 5 & 4 & 3 & . & . & 0 & .\\
7 & 7 & 6 & 5 & 4 & 3 & . & . & 0
\end{array}
$
} 
\caption{(a) This placement can be represented by \ensuremath{[\mskip1.5mu \mathrm{3},\mathrm{5},\mathrm{7},\mathrm{1},\mathrm{6},\mathrm{0},\mathrm{2},\mathrm{4}\mskip1.5mu]}. (b) Up diagonals.
(c) Down diagonals.}
\label{fig:queens-examples}
\end{figure}
} 

The aim of the puzzle is to place \ensuremath{\Varid{n}} queens on a \ensuremath{\Varid{n}} by \ensuremath{\Varid{n}} chess board such that no two queens can attack each other. Given \ensuremath{\Varid{n}}, we number the rows and columns by \ensuremath{[\mskip1.5mu \mathrm{0}\mathinner{\ldotp\ldotp}\Varid{n}\mathbin{-}\mathrm{1}\mskip1.5mu]}. Since all queens should be placed on distinct rows and distinct columns, a potential solution can be represented by a permutation \ensuremath{\Varid{xs}} of the list \ensuremath{[\mskip1.5mu \mathrm{0}\mathinner{\ldotp\ldotp}\Varid{n}\mathbin{-}\mathrm{1}\mskip1.5mu]}, such that \ensuremath{\Varid{xs}\mathbin{!!}\Varid{i}\mathrel{=}\Varid{j}} denotes that the queen on the $i$th column is placed on the $j$th row (see Figure \ref{fig:queens-examples}(a)). In this representation queens cannot be put on the same row or column, and the problem is reduced to filtering, among permutations of \ensuremath{[\mskip1.5mu \mathrm{0}\mathinner{\ldotp\ldotp}\Varid{n}\mathbin{-}\mathrm{1}\mskip1.5mu]}, those placements in which no two queens are put on the same diagonal. The specification can be written as a non-deterministic program:
\begin{hscode}\SaveRestoreHook
\column{B}{@{}>{\hspre}l<{\hspost}@{}}%
\column{E}{@{}>{\hspre}l<{\hspost}@{}}%
\>[B]{}\Varid{queens}\mathbin{::}\Conid{MonadPlus}\;\Varid{m}\Rightarrow \Conid{Int}\to \Varid{m}\;[\mskip1.5mu \Conid{Int}\mskip1.5mu]{}\<[E]%
\\
\>[B]{}\Varid{queens}\;\Varid{n}\mathrel{=}\Varid{perm}\;[\mskip1.5mu \mathrm{0}\mathinner{\ldotp\ldotp}\Varid{n}\mathbin{-}\mathrm{1}\mskip1.5mu]\mathrel{\hstretch{0.7}{>\!\!>\!\!=}}\Varid{filt}\;\Varid{safe}~~,{}\<[E]%
\ColumnHook
\end{hscode}\resethooks
where \ensuremath{\Varid{perm}} non-deterministically computes a permutation of its input, and the pure function \ensuremath{\Varid{safe}\mathbin{::}[\mskip1.5mu \Conid{Int}\mskip1.5mu]\to \Conid{Bool}} determines whether no queens are on the same diagonal.

This specification of \ensuremath{\Varid{queens}} generates all the permutations, before checking them one by one, in two separate phases. We wish to fuse the two phases and produce a faster implementation. The overall idea is to define \ensuremath{\Varid{perm}} in terms of an unfold, transform \ensuremath{\Varid{filt}\;\Varid{safe}} into a fold, and fuse the two phases into a {\em hylomorphism}~\cite{Meijer:91:Functional}. During the fusion, some non-safe choices can be pruned off earlier, speeding up the computation.

\paragraph{Permutation}
The monadic function \ensuremath{\Varid{perm}} can be written both as a fold or an unfold.
For this problem we choose the latter.
The function \ensuremath{\Varid{select}} non-deterministically splits a list into a pair containing one chosen element and the rest:
\begin{hscode}\SaveRestoreHook
\column{B}{@{}>{\hspre}l<{\hspost}@{}}%
\column{16}{@{}>{\hspre}c<{\hspost}@{}}%
\column{16E}{@{}l@{}}%
\column{19}{@{}>{\hspre}l<{\hspost}@{}}%
\column{E}{@{}>{\hspre}l<{\hspost}@{}}%
\>[B]{}\Varid{select}\mathbin{::}\Conid{MonadPlus}\;\Varid{m}\Rightarrow [\mskip1.5mu \Varid{a}\mskip1.5mu]\to \Varid{m}\;(\Varid{a},[\mskip1.5mu \Varid{a}\mskip1.5mu])~~.{}\<[E]%
\\
\>[B]{}\Varid{select}\;[\mskip1.5mu \mskip1.5mu]{}\<[16]%
\>[16]{}\mathrel{=}{}\<[16E]%
\>[19]{}\emptyset{}\<[E]%
\\
\>[B]{}\Varid{select}\;(\Varid{x}\mathbin{:}\Varid{xs}){}\<[16]%
\>[16]{}\mathrel{=}{}\<[16E]%
\>[19]{}\Varid{return}\;(\Varid{x},\Varid{xs})\mathbin{\talloblong}((\Varid{id}\times(\Varid{x}\mathbin{:}))\mathrel{\raisebox{0.5\depth}{\scaleobj{0.5}{\langle}} \scaleobj{0.8}{\$} \raisebox{0.5\depth}{\scaleobj{0.5}{\rangle}}}\Varid{select}\;\Varid{xs})~~,{}\<[E]%
\ColumnHook
\end{hscode}\resethooks
where \ensuremath{(\Varid{f}\times\Varid{g})\;(\Varid{x},\Varid{y})\mathrel{=}(\Varid{f}\;\Varid{x},\Varid{g}\;\Varid{y})}.
For example, \ensuremath{\Varid{select}\;[\mskip1.5mu \mathrm{1},\mathrm{2},\mathrm{3}\mskip1.5mu]} yields one of \ensuremath{(\mathrm{1},[\mskip1.5mu \mathrm{2},\mathrm{3}\mskip1.5mu])}, \ensuremath{(\mathrm{2},[\mskip1.5mu \mathrm{1},\mathrm{3}\mskip1.5mu])} and \ensuremath{(\mathrm{3},[\mskip1.5mu \mathrm{1},\mathrm{2}\mskip1.5mu])}. The function call \ensuremath{\Varid{unfoldM}\;\Varid{p}\;\Varid{f}\;\Varid{y}} generates a list \ensuremath{[\mskip1.5mu \Varid{a}\mskip1.5mu]} from a seed \ensuremath{\Varid{y}\mathbin{::}\Varid{b}}. If \ensuremath{\Varid{p}\;\Varid{y}} holds, the generation stops. Otherwise an element and a new seed is generated using \ensuremath{\Varid{f}}. It is like the usual \ensuremath{\Varid{unfoldr}} apart from that \ensuremath{\Varid{f}}, and thus the result, is monadic:
\begin{hscode}\SaveRestoreHook
\column{B}{@{}>{\hspre}l<{\hspost}@{}}%
\column{16}{@{}>{\hspre}l<{\hspost}@{}}%
\column{29}{@{}>{\hspre}l<{\hspost}@{}}%
\column{E}{@{}>{\hspre}l<{\hspost}@{}}%
\>[B]{}\Varid{unfoldM}\mathbin{::}\Conid{Monad}\;\Varid{m}\Rightarrow (\Varid{b}\to \Conid{Bool})\to (\Varid{b}\to \Varid{m}\;(\Varid{a},\Varid{b}))\to \Varid{b}\to \Varid{m}\;[\mskip1.5mu \Varid{a}\mskip1.5mu]{}\<[E]%
\\
\>[B]{}\Varid{unfoldM}\;\Varid{p}\;\Varid{f}\;\Varid{y}{}\<[16]%
\>[16]{}\mid \Varid{p}\;\Varid{y}{}\<[29]%
\>[29]{}\mathrel{=}\Varid{return}\;[\mskip1.5mu \mskip1.5mu]{}\<[E]%
\\
\>[16]{}\mid \Varid{otherwise}{}\<[29]%
\>[29]{}\mathrel{=}\Varid{f}\;\Varid{y}\mathrel{\hstretch{0.7}{>\!\!>\!\!=}}\lambda (\Varid{x},\Varid{z})\to (\Varid{x}\mathbin{:})\mathrel{\raisebox{0.5\depth}{\scaleobj{0.5}{\langle}} \scaleobj{0.8}{\$} \raisebox{0.5\depth}{\scaleobj{0.5}{\rangle}}}\Varid{unfoldM}\;\Varid{p}\;\Varid{f}\;\Varid{z}~~.{}\<[E]%
\ColumnHook
\end{hscode}\resethooks
Given these definitions, \ensuremath{\Varid{perm}} can be defined by:
\begin{hscode}\SaveRestoreHook
\column{B}{@{}>{\hspre}l<{\hspost}@{}}%
\column{E}{@{}>{\hspre}l<{\hspost}@{}}%
\>[B]{}\Varid{perm}\mathbin{::}\Conid{MonadPlus}\;\Varid{m}\Rightarrow [\mskip1.5mu \Varid{a}\mskip1.5mu]\to \Varid{m}\;[\mskip1.5mu \Varid{a}\mskip1.5mu]{}\<[E]%
\\
\>[B]{}\Varid{perm}\mathrel{=}\Varid{unfoldM}\;\Varid{null}\;\Varid{select}~~.{}\<[E]%
\ColumnHook
\end{hscode}\resethooks

\subsection{Safety Check in a \ensuremath{\Varid{scanl}}}

We have yet to define \ensuremath{\Varid{safe}}.
Representing a placement as a permutation allows an easy way to check whether two queens are put on the same diagonal.
An 8 by 8 chess board has 15 {\em up diagonals} (those running between bottom-left and top-right). Let them be indexed by \ensuremath{[\mskip1.5mu \mathrm{0}\mathinner{\ldotp\ldotp}\mathrm{14}\mskip1.5mu]} (see Figure \ref{fig:queens-examples}(b)).
If we apply \ensuremath{\Varid{zipWith}\;(\mathbin{+})\;[\mskip1.5mu \mathrm{0}\mathinner{\ldotp\ldotp}\mskip1.5mu]} to a permutation, we get the indices of the up-diagonals where the chess pieces are placed.
Similarly, there are 15 {\em down diagonals} (those running between top-left and bottom right).
By applying \ensuremath{\Varid{zipWith}\;(\mathbin{-})\;[\mskip1.5mu \mathrm{0}\mathinner{\ldotp\ldotp}\mskip1.5mu]} to a permutation, we get the indices of their down-diagonals (indexed by \ensuremath{[\mskip1.5mu \mathbin{-}\mathrm{7}\mathinner{\ldotp\ldotp}\mathrm{7}\mskip1.5mu]}.
See Figure \ref{fig:queens-examples}(c)).
A placement is safe if the diagonals contain no duplicates:
\begin{hscode}\SaveRestoreHook
\column{B}{@{}>{\hspre}l<{\hspost}@{}}%
\column{8}{@{}>{\hspre}l<{\hspost}@{}}%
\column{10}{@{}>{\hspre}l<{\hspost}@{}}%
\column{12}{@{}>{\hspre}l<{\hspost}@{}}%
\column{36}{@{}>{\hspre}l<{\hspost}@{}}%
\column{E}{@{}>{\hspre}l<{\hspost}@{}}%
\>[B]{}\Varid{ups},\Varid{downs}\mathbin{::}[\mskip1.5mu \Conid{Int}\mskip1.5mu]\to [\mskip1.5mu \Conid{Int}\mskip1.5mu]{}\<[E]%
\\
\>[B]{}\Varid{ups}\;{}\<[8]%
\>[8]{}\Varid{xs}{}\<[12]%
\>[12]{}\mathrel{=}\Varid{zipWith}\;(\mathbin{+})\;[\mskip1.5mu \mathrm{0}\mathinner{\ldotp\ldotp}\mskip1.5mu]\;\Varid{xs}{}\<[36]%
\>[36]{}~~,{}\<[E]%
\\
\>[B]{}\Varid{downs}\;{}\<[8]%
\>[8]{}\Varid{xs}{}\<[12]%
\>[12]{}\mathrel{=}\Varid{zipWith}\;(\mathbin{-})\;[\mskip1.5mu \mathrm{0}\mathinner{\ldotp\ldotp}\mskip1.5mu]\;\Varid{xs}{}\<[36]%
\>[36]{}~~,{}\<[E]%
\\[\blanklineskip]%
\>[B]{}\Varid{safe}{}\<[10]%
\>[10]{}\mathbin{::}[\mskip1.5mu \Conid{Int}\mskip1.5mu]\to \Conid{Bool}{}\<[E]%
\\
\>[B]{}\Varid{safe}\;\Varid{xs}{}\<[10]%
\>[10]{}\mathrel{=}\Varid{nodup}\;(\Varid{ups}\;\Varid{xs})\mathrel{\wedge}\Varid{nodup}\;(\Varid{downs}\;\Varid{xs})~~,{}\<[E]%
\ColumnHook
\end{hscode}\resethooks
where \ensuremath{\Varid{nodup}\mathbin{::}\Conid{Eq}\;\Varid{a}\Rightarrow [\mskip1.5mu \Varid{a}\mskip1.5mu]\to \Conid{Bool}} determines whether there is no duplication in a list.

The eventual goal is to transform \ensuremath{\Varid{filt}\;\Varid{safe}} into a \ensuremath{\Varid{foldr}}, to be fused with \ensuremath{\Varid{perm}}, an unfold that generates a list from left to right.
In order to do so, it helps if \ensuremath{\Varid{safe}} can be expressed in a computation that processes the list left-to-right, that is, a \ensuremath{\Varid{foldl}} or a \ensuremath{\Varid{scanl}}.
To derive such a definition we use the standard trick --- introducing accumulating parameters, and generalising \ensuremath{\Varid{safe}} to \ensuremath{\Varid{safeAcc}} below:
\begin{hscode}\SaveRestoreHook
\column{B}{@{}>{\hspre}l<{\hspost}@{}}%
\column{3}{@{}>{\hspre}l<{\hspost}@{}}%
\column{10}{@{}>{\hspre}l<{\hspost}@{}}%
\column{15}{@{}>{\hspre}l<{\hspost}@{}}%
\column{25}{@{}>{\hspre}l<{\hspost}@{}}%
\column{39}{@{}>{\hspre}l<{\hspost}@{}}%
\column{E}{@{}>{\hspre}l<{\hspost}@{}}%
\>[B]{}\Varid{safeAcc}\mathbin{::}(\Conid{Int},[\mskip1.5mu \Conid{Int}\mskip1.5mu],[\mskip1.5mu \Conid{Int}\mskip1.5mu])\to [\mskip1.5mu \Conid{Int}\mskip1.5mu]\to \Conid{Bool}{}\<[E]%
\\
\>[B]{}\Varid{safeAcc}\;(\Varid{i},\Varid{us},\Varid{ds})\;\Varid{xs}\mathrel{=}{}\<[25]%
\>[25]{}\Varid{nodup}\;\Varid{us'}\mathrel{\wedge}{}\<[39]%
\>[39]{}\Varid{nodup}\;\Varid{ds'}\mathrel{\wedge}{}\<[E]%
\\
\>[25]{}\Varid{all}\;(\not\in\Varid{us})\;\Varid{us'}\mathrel{\wedge}\Varid{all}\;(\not\in\Varid{ds})\;\Varid{ds'}~~,{}\<[E]%
\\
\>[B]{}\hsindent{3}{}\<[3]%
\>[3]{}\mathbf{where}\;{}\<[10]%
\>[10]{}\Varid{us'}{}\<[15]%
\>[15]{}\mathrel{=}\Varid{zipWith}\;(\mathbin{+})\;[\mskip1.5mu \Varid{i}\mathinner{\ldotp\ldotp}\mskip1.5mu]\;\Varid{xs}{}\<[E]%
\\
\>[10]{}\Varid{ds'}{}\<[15]%
\>[15]{}\mathrel{=}\Varid{zipWith}\;(\mathbin{-})\;[\mskip1.5mu \Varid{i}\mathinner{\ldotp\ldotp}\mskip1.5mu]\;\Varid{xs}~~.{}\<[E]%
\ColumnHook
\end{hscode}\resethooks
It is a generalisation because \ensuremath{\Varid{safe}\mathrel{=}\Varid{safeAcc}\;(\mathrm{0},[\mskip1.5mu \mskip1.5mu],[\mskip1.5mu \mskip1.5mu])}.
By plain functional calculation, one may conclude that \ensuremath{\Varid{safeAcc}} can be defined using a variation of \ensuremath{\Varid{scanl}}:
\begin{hscode}\SaveRestoreHook
\column{B}{@{}>{\hspre}l<{\hspost}@{}}%
\column{3}{@{}>{\hspre}l<{\hspost}@{}}%
\column{10}{@{}>{\hspre}l<{\hspost}@{}}%
\column{31}{@{}>{\hspre}l<{\hspost}@{}}%
\column{E}{@{}>{\hspre}l<{\hspost}@{}}%
\>[B]{}\Varid{safeAcc}\;(\Varid{i},\Varid{us},\Varid{ds})\mathrel{=}\Varid{all}\;\Varid{ok}\mathbin{\cdot}\Varid{scanl}_{+}\;(\oplus)\;(\Varid{i},\Varid{us},\Varid{ds})~~,{}\<[E]%
\\
\>[B]{}\hsindent{3}{}\<[3]%
\>[3]{}\mathbf{where}\;{}\<[10]%
\>[10]{}(\Varid{i},\Varid{us},\Varid{ds})\mathbin{\oplus}\Varid{x}{}\<[31]%
\>[31]{}\mathrel{=}(\Varid{i}\mathbin{+}\mathrm{1},(\Varid{i}\mathbin{+}\Varid{x}\mathbin{:}\Varid{us}),(\Varid{i}\mathbin{-}\Varid{x}\mathbin{:}\Varid{ds})){}\<[E]%
\\
\>[10]{}\Varid{ok}\;(\Varid{i},(\Varid{x}\mathbin{:}\Varid{us}),(\Varid{y}\mathbin{:}\Varid{ds}))\mathrel{=}\Varid{x}\not\in\Varid{us}\mathrel{\wedge}\Varid{y}\not\in\Varid{ds}~~,{}\<[E]%
\ColumnHook
\end{hscode}\resethooks
where \ensuremath{\Varid{all}\;\Varid{p}\mathrel{=}\Varid{foldr}\;(\mathrel{\wedge})\;\Conid{True}\mathbin{\cdot}\Varid{map}\;\Varid{p}} and \ensuremath{\Varid{scanl}_{+}} is like the standard  \ensuremath{\Varid{scanl}}, but applies \ensuremath{\Varid{foldl}} to all non-empty prefixes of a list.
It can be specified by:
\begin{hscode}\SaveRestoreHook
\column{B}{@{}>{\hspre}l<{\hspost}@{}}%
\column{E}{@{}>{\hspre}l<{\hspost}@{}}%
\>[B]{}\Varid{scanl}_{+}\mathbin{::}(\Varid{b}\to \Varid{a}\to \Varid{b})\to \Varid{b}\to [\mskip1.5mu \Varid{a}\mskip1.5mu]\to [\mskip1.5mu \Varid{b}\mskip1.5mu]{}\<[E]%
\\
\>[B]{}\Varid{scanl}_{+}\;(\oplus)\;\Varid{st}\mathrel{=}\Varid{tail}\mathbin{\cdot}\Varid{scanl}\;(\oplus)\;\Varid{st}~~,{}\<[E]%
\ColumnHook
\end{hscode}\resethooks
and it also adopts an inductive definition:
\begin{hscode}\SaveRestoreHook
\column{B}{@{}>{\hspre}l<{\hspost}@{}}%
\column{25}{@{}>{\hspre}l<{\hspost}@{}}%
\column{E}{@{}>{\hspre}l<{\hspost}@{}}%
\>[B]{}\Varid{scanl}_{+}\;(\oplus)\;\Varid{st}\;[\mskip1.5mu \mskip1.5mu]{}\<[25]%
\>[25]{}\mathrel{=}[\mskip1.5mu \mskip1.5mu]{}\<[E]%
\\
\>[B]{}\Varid{scanl}_{+}\;(\oplus)\;\Varid{st}\;(\Varid{x}\mathbin{:}\Varid{xs}){}\<[25]%
\>[25]{}\mathrel{=}(\Varid{st}\mathbin{\oplus}\Varid{x})\mathbin{:}\Varid{scanl}_{+}\;(\oplus)\;(\Varid{st}\mathbin{\oplus}\Varid{x})\;\Varid{xs}~~.{}\<[E]%
\ColumnHook
\end{hscode}\resethooks

Operationally, \ensuremath{\Varid{safeAcc}} examines the list from left to right, while keeping a state \ensuremath{(\Varid{i},\Varid{us},\Varid{ds})}, where \ensuremath{\Varid{i}} is the current position being examined, while \ensuremath{\Varid{us}} and \ensuremath{\Varid{ds}} are respectively indices of all the up and down diagonals encountered so far. Indeed, in a function call \ensuremath{\Varid{scanl}_{+}\;(\oplus)\;\Varid{st}}, the value \ensuremath{\Varid{st}} can be seen as a ``state'' that is explicitly carried around. This naturally leads to the idea: can we convert a \ensuremath{\Varid{scanl}_{+}} to a monadic program that stores \ensuremath{\Varid{st}} in its state? This is the goal of the next section.

As a summary of this section, after defining \ensuremath{\Varid{queens}}, we have transformed it into the following form:
\begin{hscode}\SaveRestoreHook
\column{B}{@{}>{\hspre}l<{\hspost}@{}}%
\column{3}{@{}>{\hspre}l<{\hspost}@{}}%
\column{52}{@{}>{\hspre}l<{\hspost}@{}}%
\column{E}{@{}>{\hspre}l<{\hspost}@{}}%
\>[3]{}\Varid{unfoldM}\;\Varid{p}\;\Varid{f}\mathrel{\hstretch{0.7}{>\!\!=\!\!\!>}}\Varid{filt}\;(\Varid{all}\;\Varid{ok}\mathbin{\cdot}\Varid{scanl}_{+}\;(\oplus)\;\Varid{st}){}\<[52]%
\>[52]{}~~.{}\<[E]%
\ColumnHook
\end{hscode}\resethooks
This is the form of problems we will consider for the rest of this report: problems whose solutions are generated by an monadic unfold, before being filtered by an \ensuremath{\Varid{filt}} that takes the result of a \ensuremath{\Varid{scanl}_{+}}.

\section{From Pure to Stateful \ensuremath{\Varid{scanl}}}
\label{sec:monadic-scanl}

The aim of this section is to turn the filtering phase \ensuremath{\Varid{filt}\;(\Varid{all}\;\Varid{ok}\mathbin{\cdot}\Varid{scanl}_{+}\;(\oplus)\;\Varid{st})} into a \ensuremath{\Varid{foldr}}. For that we introduce a state monad to pass the state around.

The state effect provides two operators:
\begin{hscode}\SaveRestoreHook
\column{B}{@{}>{\hspre}l<{\hspost}@{}}%
\column{5}{@{}>{\hspre}l<{\hspost}@{}}%
\column{10}{@{}>{\hspre}l<{\hspost}@{}}%
\column{E}{@{}>{\hspre}l<{\hspost}@{}}%
\>[B]{}\mathbf{class}\;\Conid{Monad}\;\Varid{m}\Rightarrow \Conid{MonadState}\;\Varid{s}\;\Varid{m}\;\mathbf{where}{}\<[E]%
\\
\>[B]{}\hsindent{5}{}\<[5]%
\>[5]{}\Varid{get}{}\<[10]%
\>[10]{}\mathbin{::}\Varid{m}\;\Varid{s}{}\<[E]%
\\
\>[B]{}\hsindent{5}{}\<[5]%
\>[5]{}\Varid{put}{}\<[10]%
\>[10]{}\mathbin{::}\Varid{s}\to \Varid{m}\;()~~,{}\<[E]%
\ColumnHook
\end{hscode}\resethooks
where \ensuremath{\Varid{get}} retrieves the state, while \ensuremath{\Varid{put}} overwrites the state by the given value. They are supposed to satisfy the \emph{state laws}:
\begin{alignat}{2}
&\mbox{\bf put-put}:\quad &
\ensuremath{\Varid{put}\;\Varid{st}\mathbin{\hstretch{0.7}{>\!\!>}}\Varid{put}\;\Varid{st'}} &= \ensuremath{\Varid{put}\;\Varid{st'}}~~\mbox{,} \label{eq:put-put}\\
&\mbox{\bf put-get}:~ &
\ensuremath{\Varid{put}\;\Varid{st}\mathbin{\hstretch{0.7}{>\!\!>}}\Varid{get}} &= \ensuremath{\Varid{put}\;\Varid{st}\mathbin{\hstretch{0.7}{>\!\!>}}\Varid{return}\;\Varid{st}} ~~\mbox{,} \label{eq:get-put}\\
&\mbox{\bf get-put}:~ &
\ensuremath{\Varid{get}\mathrel{\hstretch{0.7}{>\!\!>\!\!=}}\Varid{put}} &= \ensuremath{\Varid{return}\;()} ~~\mbox{,} \label{eq:put-get}\\
&\mbox{\bf get-get}:\quad &
\ensuremath{\Varid{get}\mathrel{\hstretch{0.7}{>\!\!>\!\!=}}(\lambda \Varid{st}\to \Varid{get}\mathrel{\hstretch{0.7}{>\!\!>\!\!=}}\Varid{k}\;\Varid{st})} &= \ensuremath{\Varid{get}\mathrel{\hstretch{0.7}{>\!\!>\!\!=}}(\lambda \Varid{st}\to \Varid{k}\;\Varid{st}\;\Varid{st})}
~~\mbox{.} \label{eq:get-get}
\end{alignat}

\subsection{From \ensuremath{\Varid{scanl}_{+}} to monadic \ensuremath{\Varid{foldr}}}
\label{sec:scanl-scanlM}

Consider the following monadic variation of \ensuremath{\Varid{scanl}}:
\begin{hscode}\SaveRestoreHook
\column{B}{@{}>{\hspre}l<{\hspost}@{}}%
\column{3}{@{}>{\hspre}l<{\hspost}@{}}%
\column{25}{@{}>{\hspre}l<{\hspost}@{}}%
\column{E}{@{}>{\hspre}l<{\hspost}@{}}%
\>[B]{}\Varid{scanlM}\mathbin{::}\Conid{MonadState}\;\Varid{s}\;\Varid{m}\Rightarrow (\Varid{s}\to \Varid{a}\to \Varid{s})\to \Varid{s}\to [\mskip1.5mu \Varid{a}\mskip1.5mu]\to \Varid{m}\;[\mskip1.5mu \Varid{s}\mskip1.5mu]{}\<[E]%
\\
\>[B]{}\Varid{scanlM}\;(\oplus)\;\Varid{st}\;\Varid{xs}\mathrel{=}\Varid{put}\;\Varid{st}\mathbin{\hstretch{0.7}{>\!\!>}}\Varid{foldr}\;(\otimes)\;(\Varid{return}\;[\mskip1.5mu \mskip1.5mu])\;\Varid{xs}{}\<[E]%
\\
\>[B]{}\hsindent{3}{}\<[3]%
\>[3]{}\mathbf{where}\;\Varid{x}\mathbin{\otimes}\Varid{n}\mathrel{=}{}\<[25]%
\>[25]{}\Varid{get}\mathrel{\hstretch{0.7}{>\!\!>\!\!=}}\lambda \Varid{st}\to \mathbf{let}\;\Varid{st'}\mathrel{=}\Varid{st}\mathbin{\oplus}\Varid{x}{}\<[E]%
\\
\>[25]{}\mathbf{in}\;(\Varid{st'}\mathbin{:})\mathrel{\raisebox{0.5\depth}{\scaleobj{0.5}{\langle}} \scaleobj{0.8}{\$} \raisebox{0.5\depth}{\scaleobj{0.5}{\rangle}}}(\Varid{put}\;\Varid{st'}\mathbin{\hstretch{0.7}{>\!\!>}}\Varid{n})~~.{}\<[E]%
\ColumnHook
\end{hscode}\resethooks
It behaves like \ensuremath{\Varid{scanl}_{+}}, but stores the accumulated information in a monadic state, which is retrieved and stored in each step. The main body of the computation is implemented using a \ensuremath{\Varid{foldr}}.

To relate \ensuremath{\Varid{scanl}_{+}} and \ensuremath{\Varid{scanlM}}, one would like to have \ensuremath{\Varid{return}\;(\Varid{scanl}_{+}\;(\oplus)\;\Varid{st}\;\Varid{xs})\mathrel{=}\Varid{scanlM}\;(\oplus)\;\Varid{st}\;\Varid{xs}}.
However, the lefthand side does not alter the state, while the righthand side does.
One of the ways to make the equality hold is to manually backup and restore the state.
Define
\begin{hscode}\SaveRestoreHook
\column{B}{@{}>{\hspre}l<{\hspost}@{}}%
\column{E}{@{}>{\hspre}l<{\hspost}@{}}%
\>[B]{}\Varid{protect}\mathbin{::}\Conid{MonadState}\;\Varid{s}\;\Varid{m}\Rightarrow \Varid{m}\;\Varid{b}\to \Varid{m}\;\Varid{b}{}\<[E]%
\\
\>[B]{}\Varid{protect}\;\Varid{n}~\mathrel{=}~\Varid{get}\mathrel{\hstretch{0.7}{>\!\!>\!\!=}}\lambda \Varid{ini}\to \Varid{n}\mathrel{\hstretch{0.7}{>\!\!>\!\!=}}\lambda \Varid{x}\to \Varid{put}\;\Varid{ini}\mathbin{\hstretch{0.7}{>\!\!>}}\Varid{return}\;\Varid{x}~~,{}\<[E]%
\ColumnHook
\end{hscode}\resethooks
We have
\begin{theorem}\label{lma:scanl-loop}
For all \ensuremath{(\oplus)\mathbin{::}(\Varid{s}\to \Varid{a}\to \Varid{s})}, \ensuremath{\Varid{st}\mathbin{::}\Varid{s}}, and \ensuremath{\Varid{xs}\mathbin{::}[\mskip1.5mu \Varid{a}\mskip1.5mu]},
\begin{hscode}\SaveRestoreHook
\column{B}{@{}>{\hspre}l<{\hspost}@{}}%
\column{3}{@{}>{\hspre}l<{\hspost}@{}}%
\column{E}{@{}>{\hspre}l<{\hspost}@{}}%
\>[3]{}\Varid{return}\;(\Varid{scanl}_{+}\;(\oplus)\;\Varid{st}\;\Varid{xs})\mathbin{=}\Varid{protect}\;(\Varid{scanlM}\;(\oplus)\;\Varid{st}\;\Varid{xs})~~.{}\<[E]%
\ColumnHook
\end{hscode}\resethooks
\end{theorem}
\begin{proof} By induction on \ensuremath{\Varid{xs}}. We present the case \ensuremath{\Varid{xs}\mathbin{:=}\Varid{x}\mathbin{:}\Varid{xs}}.
\begin{hscode}\SaveRestoreHook
\column{B}{@{}>{\hspre}l<{\hspost}@{}}%
\column{7}{@{}>{\hspre}l<{\hspost}@{}}%
\column{9}{@{}>{\hspre}l<{\hspost}@{}}%
\column{10}{@{}>{\hspre}l<{\hspost}@{}}%
\column{22}{@{}>{\hspre}l<{\hspost}@{}}%
\column{E}{@{}>{\hspre}l<{\hspost}@{}}%
\>[7]{}\Varid{protect}\;(\Varid{scanlM}\;(\oplus)\;\Varid{st}\;(\Varid{x}\mathbin{:}\Varid{xs})){}\<[E]%
\\
\>[B]{}\mathbin{=}{}\<[9]%
\>[9]{}\mbox{\commentbegin  expanding definitions, let \ensuremath{\Varid{st'}\mathrel{=}\Varid{st}\mathbin{\oplus}\Varid{x}}  \commentend}{}\<[E]%
\\
\>[B]{}\hsindent{7}{}\<[7]%
\>[7]{}\Varid{get}\mathrel{\hstretch{0.7}{>\!\!>\!\!=}}\lambda \Varid{ini}\to \Varid{put}\;\Varid{st}\mathbin{\hstretch{0.7}{>\!\!>}}\Varid{get}\mathrel{\hstretch{0.7}{>\!\!>\!\!=}}\lambda \Varid{st}\to {}\<[E]%
\\
\>[B]{}\hsindent{7}{}\<[7]%
\>[7]{}((\Varid{st'}\mathbin{:})\mathrel{\raisebox{0.5\depth}{\scaleobj{0.5}{\langle}} \scaleobj{0.8}{\$} \raisebox{0.5\depth}{\scaleobj{0.5}{\rangle}}}(\Varid{put}\;\Varid{st'}\mathbin{\hstretch{0.7}{>\!\!>}}\Varid{foldr}\;(\otimes)\;(\Varid{return}\;[\mskip1.5mu \mskip1.5mu])\;\Varid{xs}))\mathrel{\hstretch{0.7}{>\!\!>\!\!=}}\lambda \Varid{r}\to {}\<[E]%
\\
\>[B]{}\hsindent{7}{}\<[7]%
\>[7]{}\Varid{put}\;\Varid{ini}\mathbin{\hstretch{0.7}{>\!\!>}}\Varid{return}\;\Varid{r}{}\<[E]%
\\
\>[B]{}\mathbin{=}{}\<[9]%
\>[9]{}\mbox{\commentbegin  by \ensuremath{\Varid{put}}-\ensuremath{\Varid{get}} \eqref{eq:get-put}  \commentend}{}\<[E]%
\\
\>[B]{}\hsindent{7}{}\<[7]%
\>[7]{}\Varid{get}\mathrel{\hstretch{0.7}{>\!\!>\!\!=}}\lambda \Varid{ini}\to \Varid{put}\;\Varid{st}\mathbin{\hstretch{0.7}{>\!\!>}}{}\<[E]%
\\
\>[B]{}\hsindent{7}{}\<[7]%
\>[7]{}((\Varid{st'}\mathbin{:})\mathrel{\raisebox{0.5\depth}{\scaleobj{0.5}{\langle}} \scaleobj{0.8}{\$} \raisebox{0.5\depth}{\scaleobj{0.5}{\rangle}}}(\Varid{put}\;\Varid{st'}\mathbin{\hstretch{0.7}{>\!\!>}}\Varid{foldr}\;(\otimes)\;(\Varid{return}\;[\mskip1.5mu \mskip1.5mu])\;\Varid{xs}))\mathrel{\hstretch{0.7}{>\!\!>\!\!=}}\lambda \Varid{r}\to {}\<[E]%
\\
\>[B]{}\hsindent{7}{}\<[7]%
\>[7]{}\Varid{put}\;\Varid{ini}\mathbin{\hstretch{0.7}{>\!\!>}}\Varid{return}\;\Varid{r}{}\<[E]%
\\
\>[B]{}\mathbin{=}{}\<[9]%
\>[9]{}\mbox{\commentbegin  by \eqref{eq:ap-bind-ap}  \commentend}{}\<[E]%
\\
\>[B]{}\hsindent{7}{}\<[7]%
\>[7]{}(\Varid{st'}\mathbin{:})\mathrel{\raisebox{0.5\depth}{\scaleobj{0.5}{\langle}} \scaleobj{0.8}{\$} \raisebox{0.5\depth}{\scaleobj{0.5}{\rangle}}}({}\<[22]%
\>[22]{}\Varid{get}\mathrel{\hstretch{0.7}{>\!\!>\!\!=}}\lambda \Varid{ini}\to \Varid{put}\;\Varid{st}\mathbin{\hstretch{0.7}{>\!\!>}}\Varid{put}\;\Varid{st'}\mathbin{\hstretch{0.7}{>\!\!>}}{}\<[E]%
\\
\>[22]{}\Varid{foldr}\;(\otimes)\;(\Varid{return}\;[\mskip1.5mu \mskip1.5mu])\;\Varid{xs}\mathrel{\hstretch{0.7}{>\!\!>\!\!=}}\lambda \Varid{r}\to {}\<[E]%
\\
\>[22]{}\Varid{put}\;\Varid{ini}\mathbin{\hstretch{0.7}{>\!\!>}}\Varid{return}\;\Varid{r}){}\<[E]%
\\
\>[B]{}\mathbin{=}{}\<[10]%
\>[10]{}\mbox{\commentbegin  by \ensuremath{\Varid{put}}-\ensuremath{\Varid{put}} \eqref{eq:put-put}  \commentend}{}\<[E]%
\\
\>[B]{}\hsindent{7}{}\<[7]%
\>[7]{}(\Varid{st'}\mathbin{:})\mathrel{\raisebox{0.5\depth}{\scaleobj{0.5}{\langle}} \scaleobj{0.8}{\$} \raisebox{0.5\depth}{\scaleobj{0.5}{\rangle}}}({}\<[22]%
\>[22]{}\Varid{get}\mathrel{\hstretch{0.7}{>\!\!>\!\!=}}\lambda \Varid{ini}\to \Varid{put}\;\Varid{st'}\mathbin{\hstretch{0.7}{>\!\!>}}\Varid{foldr}\;(\otimes)\;(\Varid{return}\;[\mskip1.5mu \mskip1.5mu])\;\Varid{xs}{}\<[E]%
\\
\>[22]{}\mathrel{\hstretch{0.7}{>\!\!>\!\!=}}\lambda \Varid{r}\to \Varid{put}\;\Varid{ini}\mathbin{\hstretch{0.7}{>\!\!>}}\Varid{return}\;\Varid{r}){}\<[E]%
\\
\>[B]{}\mathbin{=}{}\<[9]%
\>[9]{}\mbox{\commentbegin  definitions of \ensuremath{\Varid{scanlM}} and \ensuremath{\Varid{protect}}  \commentend}{}\<[E]%
\\
\>[B]{}\hsindent{7}{}\<[7]%
\>[7]{}(\Varid{st'}\mathbin{:})\mathrel{\raisebox{0.5\depth}{\scaleobj{0.5}{\langle}} \scaleobj{0.8}{\$} \raisebox{0.5\depth}{\scaleobj{0.5}{\rangle}}}\Varid{protect}\;(\Varid{scanlM}\;(\oplus)\;\Varid{st'}\;\Varid{xs}){}\<[E]%
\\
\>[B]{}\mathbin{=}{}\<[9]%
\>[9]{}\mbox{\commentbegin  induction  \commentend}{}\<[E]%
\\
\>[B]{}\hsindent{7}{}\<[7]%
\>[7]{}(\Varid{st'}\mathbin{:})\mathrel{\raisebox{0.5\depth}{\scaleobj{0.5}{\langle}} \scaleobj{0.8}{\$} \raisebox{0.5\depth}{\scaleobj{0.5}{\rangle}}}\Varid{return}\;(\Varid{scanl}_{+}\;(\oplus)\;\Varid{st'}\;\Varid{xs}){}\<[E]%
\\
\>[B]{}\mathbin{=}{}\<[7]%
\>[7]{}\Varid{return}\;((\Varid{st}\mathbin{\oplus}\Varid{x})\mathbin{:}\Varid{scanl}_{+}\;(\oplus)\;(\Varid{st}\mathbin{\oplus}\Varid{x})\;\Varid{xs}){}\<[E]%
\\
\>[B]{}\mathbin{=}{}\<[7]%
\>[7]{}\Varid{return}\;(\Varid{scanl}_{+}\;(\oplus)\;\Varid{st}\;(\Varid{x}\mathbin{:}\Varid{xs}))~~.{}\<[E]%
\ColumnHook
\end{hscode}\resethooks
\end{proof}
This proof is instructive due to the use of properties \eqref{eq:put-put} and \eqref{eq:get-put}, and that \ensuremath{(\Varid{st'}\mathbin{:})}, being a pure function, can be easily moved around using \eqref{eq:ap-bind-ap}.

We have learned that \ensuremath{\Varid{scanl}_{+}\;(\oplus)\;\Varid{st}} can be turned into \ensuremath{\Varid{scanlM}\;(\oplus)\;\Varid{st}}, defined in terms of a stateful \ensuremath{\Varid{foldr}}.
In the definition, state is the only effect involved.
The next task is to transform \ensuremath{\Varid{filt}\;(\Varid{scanl}_{+}\;(\oplus)\;\Varid{st})} into a \ensuremath{\Varid{foldr}}.
The operator \ensuremath{\Varid{filt}} is defined using non-determinism.
Hence the transformation involves the interaction between two effects.

\subsection{Right-Distributivity and Local State}
\label{sec:right-distr-local-state}

We now digress a little to discuss one form of interaction between non-determinism and state.
In this report, we wish that the following two additional properties are valid:
\begin{alignat}{2}
&\mbox{\bf right-distributivity}:\quad&
  \ensuremath{\Varid{m}\mathrel{\hstretch{0.7}{>\!\!>\!\!=}}(\lambda \Varid{x}\to \Varid{f}_{1}\;\Varid{x}\mathbin{\talloblong}\Varid{f}_{2}\;\Varid{x})}~ &=~ \ensuremath{(\Varid{m}\mathrel{\hstretch{0.7}{>\!\!>\!\!=}}\Varid{f}_{1})\mathbin{\talloblong}(\Varid{m}\mathrel{\hstretch{0.7}{>\!\!>\!\!=}}\Varid{f}_{2})} \mbox{~~,}
    \label{eq:mplus-bind-dist}\\
&\mbox{\bf right-zero}:\quad&
  \ensuremath{\Varid{m}\mathbin{\hstretch{0.7}{>\!\!>}}\emptyset}~ &=~ \ensuremath{\emptyset} ~~\mbox{~~.}
    \label{eq:mzero-bind-zero}
\end{alignat}
Note that the two properties hold for some monads with non-determinism, but not all.
With some implementations of the monad, it is likely that in the lefthand side of \eqref{eq:mplus-bind-dist}, the effect of \ensuremath{\Varid{m}} happens once, while in the righthand side it happens twice. In \eqref{eq:mzero-bind-zero}, the \ensuremath{\Varid{m}} on the lefthand side may incur some effects that do not happen in the righthand side.

Having \eqref{eq:mplus-bind-dist} and \eqref{eq:mzero-bind-zero} leads to profound consequences on the semantics and implementation of monadic programs.
To begin with, \eqref{eq:mplus-bind-dist} implies that \ensuremath{(\talloblong)} be commutative. To see that, let \ensuremath{\Varid{m}\mathrel{=}\Varid{m}_{1}\mathbin{\talloblong}\Varid{m}_{2}} and \ensuremath{\Varid{f}_{1}\mathrel{=}\Varid{f}_{2}\mathrel{=}\Varid{return}} in \eqref{eq:mplus-bind-dist}.
Implementation of such non-deterministic monads have been studied by Kiselyov~\shortcite{Kiselyov:13:How}.

When mixed with state, one consequence of \eqref{eq:mplus-bind-dist} is that \ensuremath{\Varid{get}\mathrel{\hstretch{0.7}{>\!\!>\!\!=}}(\lambda \Varid{s}\to \Varid{f}_{1}\;\Varid{s}\mathbin{\talloblong}\Varid{f}_{2}\;\Varid{s})\mathrel{=}(\Varid{get}\mathrel{\hstretch{0.7}{>\!\!>\!\!=}}\Varid{f}_{1}\mathbin{\talloblong}\Varid{get}\mathrel{\hstretch{0.7}{>\!\!>\!\!=}}\Varid{f}_{2})}. That is, \ensuremath{\Varid{f}_{1}} and \ensuremath{\Varid{f}_{2}} get the same state regardless of whether \ensuremath{\Varid{get}} is performed outside or inside the non-determinism branch.
Similarly, \eqref{eq:mzero-bind-zero} implies \ensuremath{\Varid{put}\;\Varid{s}\mathbin{\hstretch{0.7}{>\!\!>}}\emptyset\mathrel{=}\emptyset} --- when a program fails, the changes it performed on the state can be discarded.
These requirements imply that \emph{each non-determinism branch has its own copy of the state}.
Therefore, we will refer to \eqref{eq:mplus-bind-dist} and \eqref{eq:mzero-bind-zero} as \emph{local state laws} in this report --- even though they do not explicitly mention state operators at all!

One monad satisfying the local state laws is \ensuremath{\Conid{M}\;\Varid{a}\mathrel{=}\Varid{s}\to [\mskip1.5mu (\Varid{a},\Varid{s})\mskip1.5mu]}, which is the same monad one gets by \ensuremath{\Conid{StateT}\;\Varid{s}\;(\Conid{ListT}\;\Conid{Identity})} in the Monad Transformer Library~\cite{MTL:14}.
With effect handling~\cite{Wu:14:Effect, KiselyovIshii:15:Freer}, the monad meets the requirements if we run the handler for state before that for list.

The advantage of having the local state laws is that we get many useful properties, which make this stateful non-determinism monad preferred for program calculation and reasoning.
Recall, for example, that \eqref{eq:mzero-bind-zero} is the antecedent of
\eqref{eq:guard-commute}.
The result can be stronger: non-determinism commutes with all other effects if we have local state laws.
\begin{definition}
Let \ensuremath{\Varid{m}} and \ensuremath{\Varid{n}} be two monadic programs such that \ensuremath{\Varid{x}} does not occur free in \ensuremath{\Varid{m}}, and \ensuremath{\Varid{y}} does not occur free in \ensuremath{\Varid{n}}. We say \ensuremath{\Varid{m}} and \ensuremath{\Varid{n}} commute if
\begin{equation} \label{eq:commute}
\begin{split}
  \ensuremath{\Varid{m}\mathrel{\hstretch{0.7}{>\!\!>\!\!=}}\lambda \Varid{x}\to \Varid{n}\mathrel{\hstretch{0.7}{>\!\!>\!\!=}}\lambda \Varid{y}\to \Varid{f}\;\Varid{x}\;\Varid{y}~\mathrel{=}}\\
   \ensuremath{\Varid{n}\mathrel{\hstretch{0.7}{>\!\!>\!\!=}}\lambda \Varid{y}\to \Varid{m}\mathrel{\hstretch{0.7}{>\!\!>\!\!=}}\lambda \Varid{x}\to \Varid{f}\;\Varid{x}\;\Varid{y}~~.}
\end{split}
\end{equation}
We say that \ensuremath{\Varid{m}} commutes with effect \ensuremath{\delta} if \ensuremath{\Varid{m}} commutes with any \ensuremath{\Varid{n}} whose only effects are \ensuremath{\delta}, and that effects \ensuremath{\epsilon} and \ensuremath{\delta} commute if any \ensuremath{\Varid{m}} and \ensuremath{\Varid{n}} commute as long as their only effects are respectively \ensuremath{\epsilon} and \ensuremath{\delta}.
\end{definition}

\begin{theorem} \label{thm:nondet-commute}
If right-distributivity \eqref{eq:mplus-bind-dist} and right-zero \eqref{eq:mzero-bind-zero} hold
in addition to the monad laws stated before, non-determinism commutes with any effect \ensuremath{\epsilon}.
\end{theorem}
\begin{proof} Let \ensuremath{\Varid{m}} be a monadic program whose only effect is non-determinism, and \ensuremath{\Varid{stmt}} be any monadic program. The aim is to prove that \ensuremath{\Varid{m}} and \ensuremath{\Varid{stmt}} commute. Induction on the structure of \ensuremath{\Varid{m}}.

{\sc Case} \ensuremath{\Varid{m}\mathbin{:=}\Varid{return}\;\Varid{e}}:
\begin{hscode}\SaveRestoreHook
\column{B}{@{}>{\hspre}l<{\hspost}@{}}%
\column{4}{@{}>{\hspre}l<{\hspost}@{}}%
\column{9}{@{}>{\hspre}l<{\hspost}@{}}%
\column{E}{@{}>{\hspre}l<{\hspost}@{}}%
\>[4]{}\Varid{stmt}\mathrel{\hstretch{0.7}{>\!\!>\!\!=}}\lambda \Varid{x}\to \Varid{return}\;\Varid{e}\mathrel{\hstretch{0.7}{>\!\!>\!\!=}}\lambda \Varid{y}\to \Varid{f}\;\Varid{x}\;\Varid{y}{}\<[E]%
\\
\>[B]{}\mathbin{=}{}\<[9]%
\>[9]{}\mbox{\commentbegin  monad law \eqref{eq:monad-bind-ret}  \commentend}{}\<[E]%
\\
\>[B]{}\hsindent{4}{}\<[4]%
\>[4]{}\Varid{stmt}\mathrel{\hstretch{0.7}{>\!\!>\!\!=}}\lambda \Varid{x}\to \Varid{f}\;\Varid{x}\;\Varid{e}{}\<[E]%
\\
\>[B]{}\mathbin{=}{}\<[9]%
\>[9]{}\mbox{\commentbegin  monad law \eqref{eq:monad-bind-ret}  \commentend}{}\<[E]%
\\
\>[B]{}\hsindent{4}{}\<[4]%
\>[4]{}\Varid{return}\;\Varid{e}\mathrel{\hstretch{0.7}{>\!\!>\!\!=}}\lambda \Varid{y}\to \Varid{stmt}\mathrel{\hstretch{0.7}{>\!\!>\!\!=}}\lambda \Varid{x}\to \Varid{f}\;\Varid{x}\;\Varid{y}~~.{}\<[E]%
\ColumnHook
\end{hscode}\resethooks

{\sc Case} \ensuremath{\Varid{m}\mathbin{:=}\Varid{m}_{1}\mathbin{\talloblong}\Varid{m}_{2}}:
\begin{hscode}\SaveRestoreHook
\column{B}{@{}>{\hspre}l<{\hspost}@{}}%
\column{4}{@{}>{\hspre}l<{\hspost}@{}}%
\column{8}{@{}>{\hspre}l<{\hspost}@{}}%
\column{E}{@{}>{\hspre}l<{\hspost}@{}}%
\>[4]{}\Varid{stmt}\mathrel{\hstretch{0.7}{>\!\!>\!\!=}}\lambda \Varid{x}\to (\Varid{m}_{1}\mathbin{\talloblong}\Varid{m}_{2})\mathrel{\hstretch{0.7}{>\!\!>\!\!=}}\lambda \Varid{y}\to \Varid{f}\;\Varid{x}\;\Varid{y}{}\<[E]%
\\
\>[B]{}\mathbin{=}{}\<[8]%
\>[8]{}\mbox{\commentbegin  by \eqref{eq:bind-mplus-dist}  \commentend}{}\<[E]%
\\
\>[B]{}\hsindent{4}{}\<[4]%
\>[4]{}\Varid{stmt}\mathrel{\hstretch{0.7}{>\!\!>\!\!=}}\lambda \Varid{x}\to (\Varid{m}_{1}\mathrel{\hstretch{0.7}{>\!\!>\!\!=}}\Varid{f}\;\Varid{x})\mathbin{\talloblong}(\Varid{m}_{2}\mathrel{\hstretch{0.7}{>\!\!>\!\!=}}\Varid{f}\;\Varid{x}){}\<[E]%
\\
\>[B]{}\mathbin{=}{}\<[8]%
\>[8]{}\mbox{\commentbegin  by \eqref{eq:mplus-bind-dist}  \commentend}{}\<[E]%
\\
\>[B]{}\hsindent{4}{}\<[4]%
\>[4]{}(\Varid{stmt}\mathrel{\hstretch{0.7}{>\!\!>\!\!=}}\lambda \Varid{x}\to \Varid{m}_{1}\mathrel{\hstretch{0.7}{>\!\!>\!\!=}}\Varid{f}\;\Varid{x})\mathbin{\talloblong}(\Varid{stmt}\mathrel{\hstretch{0.7}{>\!\!>\!\!=}}\lambda \Varid{x}\to \Varid{m}_{2}\mathrel{\hstretch{0.7}{>\!\!>\!\!=}}\Varid{f}\;\Varid{x}){}\<[E]%
\\
\>[B]{}\mathbin{=}{}\<[8]%
\>[8]{}\mbox{\commentbegin  induction  \commentend}{}\<[E]%
\\
\>[B]{}\hsindent{4}{}\<[4]%
\>[4]{}(\Varid{m}_{1}\mathrel{\hstretch{0.7}{>\!\!>\!\!=}}\lambda \Varid{y}\to \Varid{stmt}\mathrel{\hstretch{0.7}{>\!\!>\!\!=}}\lambda \Varid{x}\to \Varid{f}\;\Varid{x}\;\Varid{y})\mathbin{\talloblong}(\Varid{m}_{2}\mathrel{\hstretch{0.7}{>\!\!>\!\!=}}\lambda \Varid{y}\to \Varid{stmt}\mathrel{\hstretch{0.7}{>\!\!>\!\!=}}\lambda \Varid{x}\to \Varid{f}\;\Varid{x}\;\Varid{y}){}\<[E]%
\\
\>[B]{}\mathbin{=}{}\<[8]%
\>[8]{}\mbox{\commentbegin  by \eqref{eq:bind-mplus-dist}  \commentend}{}\<[E]%
\\
\>[B]{}\hsindent{4}{}\<[4]%
\>[4]{}(\Varid{m}_{1}\mathbin{\talloblong}\Varid{m}_{2})\mathrel{\hstretch{0.7}{>\!\!>\!\!=}}\lambda \Varid{y}\to \Varid{stmt}\mathrel{\hstretch{0.7}{>\!\!>\!\!=}}\lambda \Varid{x}\to \Varid{f}\;\Varid{x}\;\Varid{y}~~.{}\<[E]%
\ColumnHook
\end{hscode}\resethooks

{\sc Case} \ensuremath{\Varid{m}\mathbin{:=}\emptyset}:
\begin{hscode}\SaveRestoreHook
\column{B}{@{}>{\hspre}l<{\hspost}@{}}%
\column{4}{@{}>{\hspre}l<{\hspost}@{}}%
\column{9}{@{}>{\hspre}l<{\hspost}@{}}%
\column{E}{@{}>{\hspre}l<{\hspost}@{}}%
\>[4]{}\Varid{stmt}\mathrel{\hstretch{0.7}{>\!\!>\!\!=}}\lambda \Varid{x}\to \emptyset\mathrel{\hstretch{0.7}{>\!\!>\!\!=}}\lambda \Varid{y}\to \Varid{f}\;\Varid{x}\;\Varid{y}{}\<[E]%
\\
\>[B]{}\mathbin{=}{}\<[9]%
\>[9]{}\mbox{\commentbegin  by \eqref{eq:bind-mzero-zero}  \commentend}{}\<[E]%
\\
\>[B]{}\hsindent{4}{}\<[4]%
\>[4]{}\Varid{stmt}\mathrel{\hstretch{0.7}{>\!\!>\!\!=}}\lambda \Varid{x}\to \emptyset{}\<[E]%
\\
\>[B]{}\mathbin{=}{}\<[9]%
\>[9]{}\mbox{\commentbegin  by \eqref{eq:mzero-bind-zero}  \commentend}{}\<[E]%
\\
\>[B]{}\hsindent{4}{}\<[4]%
\>[4]{}\emptyset{}\<[E]%
\\
\>[B]{}\mathbin{=}{}\<[9]%
\>[9]{}\mbox{\commentbegin  by \eqref{eq:bind-mzero-zero}  \commentend}{}\<[E]%
\\
\>[B]{}\hsindent{4}{}\<[4]%
\>[4]{}\emptyset\mathrel{\hstretch{0.7}{>\!\!>\!\!=}}\lambda \Varid{y}\to \Varid{stmt}\mathrel{\hstretch{0.7}{>\!\!>\!\!=}}\lambda \Varid{x}\to \Varid{f}\;\Varid{x}\;\Varid{y}~~.{}\<[E]%
\ColumnHook
\end{hscode}\resethooks
\end{proof}

\paragraph{Note} We briefly justify proofs by induction on the syntax tree.
Finite monadic programs can be represented by the free monad constructed out of \ensuremath{\Varid{return}} and the effect operators, which can be represented by an inductively defined data structure, and interpreted by effect handlers~\cite{Kiselyov:13:Extensible, KiselyovIshii:15:Freer}.
When we say two programs \ensuremath{\Varid{m}_{1}} and \ensuremath{\Varid{m}_{2}} are equal, we mean that they have the same denotation when interpreted by the effect handlers of the corresponding effects, for example, \ensuremath{\Varid{hdNondet}\;(\Varid{hdState}\;\Varid{s}\;\Varid{m}_{1})\mathrel{=}\Varid{hdNondet}\;(\Varid{hdState}\;\Varid{s}\;\Varid{m}_{2})}, where \ensuremath{\Varid{hdNondet}} and \ensuremath{\Varid{hdState}} are respectively handlers for nondeterminism and state.
Such equality can be proved by induction on some sub-expression in \ensuremath{\Varid{m}_{1}} or \ensuremath{\Varid{m}_{2}}, which are treated like any inductively defined data structure.
A more complete treatment is a work in progress.
({\em End of Note})

\subsection{Filtering Using a Stateful, Non-Deterministic Fold}
\label{sec:monadic-state-passing-local}

Having dealt with \ensuremath{\Varid{scanl}_{+}\;(\oplus)\;\Varid{st}} in Section \ref{sec:scanl-scanlM},
in this section we aim to turn a filter of the form \ensuremath{\Varid{filt}\;(\Varid{all}\;\Varid{ok}\mathbin{\cdot}\Varid{scanl}_{+}\;(\oplus)\;\Varid{st})} to a stateful and non-deterministic \ensuremath{\Varid{foldr}}.

We calculate, for all \ensuremath{\Varid{ok}}, \ensuremath{(\oplus)}, \ensuremath{\Varid{st}}, and \ensuremath{\Varid{xs}}:
\begin{hscode}\SaveRestoreHook
\column{B}{@{}>{\hspre}l<{\hspost}@{}}%
\column{7}{@{}>{\hspre}l<{\hspost}@{}}%
\column{9}{@{}>{\hspre}l<{\hspost}@{}}%
\column{E}{@{}>{\hspre}l<{\hspost}@{}}%
\>[7]{}\Varid{filt}\;(\Varid{all}\;\Varid{ok}\mathbin{\cdot}\Varid{scanl}_{+}\;(\oplus)\;\Varid{st})\;\Varid{xs}{}\<[E]%
\\
\>[B]{}\mathbin{=}{}\<[7]%
\>[7]{}\Varid{guard}\;(\Varid{all}\;\Varid{ok}\;(\Varid{scanl}_{+}\;(\oplus)\;\Varid{st}\;\Varid{xs}))\mathbin{\hstretch{0.7}{>\!\!>}}\Varid{return}\;\Varid{xs}{}\<[E]%
\\
\>[B]{}\mathbin{=}{}\<[7]%
\>[7]{}\Varid{return}\;(\Varid{scanl}_{+}\;(\oplus)\;\Varid{st}\;\Varid{xs})\mathrel{\hstretch{0.7}{>\!\!>\!\!=}}\lambda \Varid{ys}\to {}\<[E]%
\\
\>[7]{}\Varid{guard}\;(\Varid{all}\;\Varid{ok}\;\Varid{ys})\mathbin{\hstretch{0.7}{>\!\!>}}\Varid{return}\;\Varid{xs}{}\<[E]%
\\
\>[B]{}\mathbin{=}{}\<[9]%
\>[9]{}\mbox{\commentbegin  Theorem~\ref{lma:scanl-loop}, definition of \ensuremath{\Varid{protect}}, monad law  \commentend}{}\<[E]%
\\
\>[B]{}\hsindent{7}{}\<[7]%
\>[7]{}\Varid{get}\mathrel{\hstretch{0.7}{>\!\!>\!\!=}}\lambda \Varid{ini}\to \Varid{scanlM}\;(\oplus)\;\Varid{st}\;\Varid{xs}\mathrel{\hstretch{0.7}{>\!\!>\!\!=}}\lambda \Varid{ys}\to \Varid{put}\;\Varid{ini}\mathbin{\hstretch{0.7}{>\!\!>}}{}\<[E]%
\\
\>[B]{}\hsindent{7}{}\<[7]%
\>[7]{}\Varid{guard}\;(\Varid{all}\;\Varid{ok}\;\Varid{ys})\mathbin{\hstretch{0.7}{>\!\!>}}\Varid{return}\;\Varid{xs}{}\<[E]%
\\
\>[B]{}\mathbin{=}{}\<[9]%
\>[9]{}\mbox{\commentbegin  Theorem~\ref{thm:nondet-commute}: non-determinism commutes with state  \commentend}{}\<[E]%
\\
\>[B]{}\hsindent{7}{}\<[7]%
\>[7]{}\Varid{get}\mathrel{\hstretch{0.7}{>\!\!>\!\!=}}\lambda \Varid{ini}\to \Varid{scanlM}\;(\oplus)\;\Varid{st}\;\Varid{xs}\mathrel{\hstretch{0.7}{>\!\!>\!\!=}}\lambda \Varid{ys}\to {}\<[E]%
\\
\>[B]{}\hsindent{7}{}\<[7]%
\>[7]{}\Varid{guard}\;(\Varid{all}\;\Varid{ok}\;\Varid{ys})\mathbin{\hstretch{0.7}{>\!\!>}}\Varid{put}\;\Varid{ini}\mathbin{\hstretch{0.7}{>\!\!>}}\Varid{return}\;\Varid{xs}{}\<[E]%
\\
\>[B]{}\mathbin{=}{}\<[9]%
\>[9]{}\mbox{\commentbegin  definition of \ensuremath{\Varid{protect}}, monad laws  \commentend}{}\<[E]%
\\
\>[B]{}\hsindent{7}{}\<[7]%
\>[7]{}\Varid{protect}\;(\Varid{scanlM}\;(\oplus)\;\Varid{st}\;\Varid{xs}\mathrel{\hstretch{0.7}{>\!\!>\!\!=}}(\Varid{guard}\mathbin{\cdot}\Varid{all}\;\Varid{ok})\mathbin{\hstretch{0.7}{>\!\!>}}\Varid{return}\;\Varid{xs})~~.{}\<[E]%
\ColumnHook
\end{hscode}\resethooks

Recall that \ensuremath{\Varid{scanlM}\;(\oplus)\;\Varid{st}\;\Varid{xs}\mathrel{=}\Varid{put}\;\Varid{st}\mathbin{\hstretch{0.7}{>\!\!>}}\Varid{foldr}\;(\otimes)\;(\Varid{return}\;[\mskip1.5mu \mskip1.5mu])\;\Varid{xs}}.
The following theorem fuses a monadic \ensuremath{\Varid{foldr}} with a \ensuremath{\Varid{guard}} that uses its result.
\begin{theorem}\label{lma:foldr-guard-fusion}
Assume that state and non-determinism commute.
Let \ensuremath{(\otimes)} be defined as that in \ensuremath{\Varid{scanlM}} for any given \ensuremath{(\oplus)\mathbin{::}\Varid{s}\to \Varid{a}\to \Varid{s}}. We have that for all \ensuremath{\Varid{ok}\mathbin{::}\Varid{s}\to \Conid{Bool}} and \ensuremath{\Varid{xs}\mathbin{::}[\mskip1.5mu \Varid{a}\mskip1.5mu]}:
\begin{hscode}\SaveRestoreHook
\column{B}{@{}>{\hspre}l<{\hspost}@{}}%
\column{3}{@{}>{\hspre}l<{\hspost}@{}}%
\column{5}{@{}>{\hspre}l<{\hspost}@{}}%
\column{7}{@{}>{\hspre}l<{\hspost}@{}}%
\column{25}{@{}>{\hspre}l<{\hspost}@{}}%
\column{E}{@{}>{\hspre}l<{\hspost}@{}}%
\>[3]{}\Varid{foldr}\;(\otimes)\;(\Varid{return}\;[\mskip1.5mu \mskip1.5mu])\;\Varid{xs}\mathrel{\hstretch{0.7}{>\!\!>\!\!=}}(\Varid{guard}\mathbin{\cdot}\Varid{all}\;\Varid{ok})\mathbin{\hstretch{0.7}{>\!\!>}}\Varid{return}\;\Varid{xs}\mathbin{=}{}\<[E]%
\\
\>[3]{}\hsindent{4}{}\<[7]%
\>[7]{}\Varid{foldr}\;(\odot)\;(\Varid{return}\;[\mskip1.5mu \mskip1.5mu])\;\Varid{xs}~~,{}\<[E]%
\\
\>[3]{}\hsindent{2}{}\<[5]%
\>[5]{}\mathbf{where}\;\Varid{x}\mathbin{\odot}\Varid{m}\mathrel{=}{}\<[25]%
\>[25]{}\Varid{get}\mathrel{\hstretch{0.7}{>\!\!>\!\!=}}\lambda \Varid{st}\to \Varid{guard}\;(\Varid{ok}\;(\Varid{st}\mathbin{\oplus}\Varid{x}))\mathbin{\hstretch{0.7}{>\!\!>}}{}\<[E]%
\\
\>[25]{}\Varid{put}\;(\Varid{st}\mathbin{\oplus}\Varid{x})\mathbin{\hstretch{0.7}{>\!\!>}}((\Varid{x}\mathbin{:})\mathrel{\raisebox{0.5\depth}{\scaleobj{0.5}{\langle}} \scaleobj{0.8}{\$} \raisebox{0.5\depth}{\scaleobj{0.5}{\rangle}}}\Varid{m})~~.{}\<[E]%
\ColumnHook
\end{hscode}\resethooks
\end{theorem}
\begin{proof} Unfortunately we cannot use a \ensuremath{\Varid{foldr}} fusion, since \ensuremath{\Varid{xs}}
occurs free in \ensuremath{\lambda \Varid{ys}\to \Varid{guard}\;(\Varid{all}\;\Varid{ok}\;\Varid{ys})\mathbin{\hstretch{0.7}{>\!\!>}}\Varid{return}\;\Varid{xs}}. Instead we
use a simple induction on \ensuremath{\Varid{xs}}. For the case \ensuremath{\Varid{xs}\mathbin{:=}\Varid{x}\mathbin{:}\Varid{xs}}:
\begin{hscode}\SaveRestoreHook
\column{B}{@{}>{\hspre}l<{\hspost}@{}}%
\column{7}{@{}>{\hspre}l<{\hspost}@{}}%
\column{8}{@{}>{\hspre}l<{\hspost}@{}}%
\column{9}{@{}>{\hspre}l<{\hspost}@{}}%
\column{E}{@{}>{\hspre}l<{\hspost}@{}}%
\>[7]{}(\Varid{x}\mathbin{\otimes}\Varid{foldr}\;(\otimes)\;(\Varid{return}\;[\mskip1.5mu \mskip1.5mu])\;\Varid{xs})\mathrel{\hstretch{0.7}{>\!\!>\!\!=}}(\Varid{guard}\mathbin{\cdot}\Varid{all}\;\Varid{ok})\mathbin{\hstretch{0.7}{>\!\!>}}\Varid{return}\;(\Varid{x}\mathbin{:}\Varid{xs}){}\<[E]%
\\
\>[B]{}\mathbin{=}{}\<[9]%
\>[9]{}\mbox{\commentbegin  definition of \ensuremath{(\otimes)}  \commentend}{}\<[E]%
\\
\>[B]{}\hsindent{7}{}\<[7]%
\>[7]{}\Varid{get}\mathrel{\hstretch{0.7}{>\!\!>\!\!=}}\lambda \Varid{st}\to {}\<[E]%
\\
\>[B]{}\hsindent{7}{}\<[7]%
\>[7]{}(((\Varid{st}\mathbin{\oplus}\Varid{x})\mathbin{:})\mathrel{\raisebox{0.5\depth}{\scaleobj{0.5}{\langle}} \scaleobj{0.8}{\$} \raisebox{0.5\depth}{\scaleobj{0.5}{\rangle}}}(\Varid{put}\;(\Varid{st}\mathbin{\oplus}\Varid{x})\mathbin{\hstretch{0.7}{>\!\!>}}\Varid{foldr}\;(\otimes)\;(\Varid{return}\;[\mskip1.5mu \mskip1.5mu])\;\Varid{xs}))\mathrel{\hstretch{0.7}{>\!\!>\!\!=}}{}\<[E]%
\\
\>[B]{}\hsindent{7}{}\<[7]%
\>[7]{}(\Varid{guard}\mathbin{\cdot}\Varid{all}\;\Varid{ok})\mathbin{\hstretch{0.7}{>\!\!>}}\Varid{return}\;(\Varid{x}\mathbin{:}\Varid{xs}){}\<[E]%
\\
\>[B]{}\mathbin{=}{}\<[9]%
\>[9]{}\mbox{\commentbegin  monad laws, \eqref{eq:comp-bind-ap}, and \eqref{eq:ap-bind-ap}  \commentend}{}\<[E]%
\\
\>[B]{}\hsindent{7}{}\<[7]%
\>[7]{}\Varid{get}\mathrel{\hstretch{0.7}{>\!\!>\!\!=}}\lambda \Varid{st}\to \Varid{put}\;(\Varid{st}\mathbin{\oplus}\Varid{x})\mathbin{\hstretch{0.7}{>\!\!>}}{}\<[E]%
\\
\>[B]{}\hsindent{7}{}\<[7]%
\>[7]{}\Varid{foldr}\;(\otimes)\;(\Varid{return}\;[\mskip1.5mu \mskip1.5mu])\;\Varid{xs}\mathrel{\hstretch{0.7}{>\!\!>\!\!=}}\lambda \Varid{ys}\to {}\<[E]%
\\
\>[B]{}\hsindent{7}{}\<[7]%
\>[7]{}\Varid{guard}\;(\Varid{all}\;\Varid{ok}\;(\Varid{st}\mathbin{\oplus}\Varid{x}\mathbin{:}\Varid{ys}))\mathbin{\hstretch{0.7}{>\!\!>}}\Varid{return}\;(\Varid{x}\mathbin{:}\Varid{xs}){}\<[E]%
\\
\>[B]{}\mathbin{=}{}\<[8]%
\>[8]{}\mbox{\commentbegin  since \ensuremath{\Varid{guard}\;(\Varid{p}\mathrel{\wedge}\Varid{q})\mathrel{=}\Varid{guard}\;\Varid{q}\mathbin{\hstretch{0.7}{>\!\!>}}\Varid{guard}\;\Varid{p}}  \commentend}{}\<[E]%
\\
\>[B]{}\hsindent{7}{}\<[7]%
\>[7]{}\Varid{get}\mathrel{\hstretch{0.7}{>\!\!>\!\!=}}\lambda \Varid{st}\to \Varid{put}\;(\Varid{st}\mathbin{\oplus}\Varid{x})\mathbin{\hstretch{0.7}{>\!\!>}}{}\<[E]%
\\
\>[B]{}\hsindent{7}{}\<[7]%
\>[7]{}\Varid{foldr}\;(\otimes)\;(\Varid{return}\;[\mskip1.5mu \mskip1.5mu])\;\Varid{xs}\mathrel{\hstretch{0.7}{>\!\!>\!\!=}}\lambda \Varid{ys}\to {}\<[E]%
\\
\>[B]{}\hsindent{7}{}\<[7]%
\>[7]{}\Varid{guard}\;(\Varid{ok}\;(\Varid{st}\mathbin{\oplus}\Varid{x}))\mathbin{\hstretch{0.7}{>\!\!>}}\Varid{guard}\;(\Varid{all}\;\Varid{ok}\;\Varid{ys})\mathbin{\hstretch{0.7}{>\!\!>}}{}\<[E]%
\\
\>[B]{}\hsindent{7}{}\<[7]%
\>[7]{}\Varid{return}\;(\Varid{x}\mathbin{:}\Varid{xs}){}\<[E]%
\\
\>[B]{}\mathbin{=}{}\<[9]%
\>[9]{}\mbox{\commentbegin  assumption: nondeterminism commutes with state  \commentend}{}\<[E]%
\\
\>[B]{}\hsindent{7}{}\<[7]%
\>[7]{}\Varid{get}\mathrel{\hstretch{0.7}{>\!\!>\!\!=}}\lambda \Varid{st}\to \Varid{guard}\;(\Varid{ok}\;(\Varid{st}\mathbin{\oplus}\Varid{x}))\mathbin{\hstretch{0.7}{>\!\!>}}\Varid{put}\;(\Varid{st}\mathbin{\oplus}\Varid{x})\mathbin{\hstretch{0.7}{>\!\!>}}{}\<[E]%
\\
\>[B]{}\hsindent{7}{}\<[7]%
\>[7]{}\Varid{foldr}\;(\otimes)\;(\Varid{return}\;[\mskip1.5mu \mskip1.5mu])\;\Varid{xs}\mathrel{\hstretch{0.7}{>\!\!>\!\!=}}\lambda \Varid{ys}\to {}\<[E]%
\\
\>[B]{}\hsindent{7}{}\<[7]%
\>[7]{}\Varid{guard}\;(\Varid{all}\;\Varid{ok}\;\Varid{ys})\mathbin{\hstretch{0.7}{>\!\!>}}\Varid{return}\;(\Varid{x}\mathbin{:}\Varid{xs}){}\<[E]%
\\
\>[B]{}\mathbin{=}{}\<[9]%
\>[9]{}\mbox{\commentbegin  monad laws and definition of \ensuremath{(\mathrel{\raisebox{0.5\depth}{\scaleobj{0.5}{\langle}} \scaleobj{0.8}{\$} \raisebox{0.5\depth}{\scaleobj{0.5}{\rangle}}})}   \commentend}{}\<[E]%
\\
\>[B]{}\hsindent{7}{}\<[7]%
\>[7]{}\Varid{get}\mathrel{\hstretch{0.7}{>\!\!>\!\!=}}\lambda \Varid{st}\to \Varid{guard}\;(\Varid{ok}\;(\Varid{st}\mathbin{\oplus}\Varid{x}))\mathbin{\hstretch{0.7}{>\!\!>}}\Varid{put}\;(\Varid{st}\mathbin{\oplus}\Varid{x})\mathbin{\hstretch{0.7}{>\!\!>}}{}\<[E]%
\\
\>[B]{}\hsindent{7}{}\<[7]%
\>[7]{}(\Varid{x}\mathbin{:})\mathrel{\raisebox{0.5\depth}{\scaleobj{0.5}{\langle}} \scaleobj{0.8}{\$} \raisebox{0.5\depth}{\scaleobj{0.5}{\rangle}}}(\Varid{foldr}\;(\otimes)\;(\Varid{return}\;[\mskip1.5mu \mskip1.5mu])\;\Varid{xs}\mathrel{\hstretch{0.7}{>\!\!>\!\!=}}\lambda \Varid{ys}\to \Varid{guard}\;(\Varid{all}\;\Varid{ok}\;\Varid{ys})\mathbin{\hstretch{0.7}{>\!\!>}}\Varid{return}\;\Varid{xs}){}\<[E]%
\\
\>[B]{}\mathbin{=}{}\<[9]%
\>[9]{}\mbox{\commentbegin  induction  \commentend}{}\<[E]%
\\
\>[B]{}\hsindent{7}{}\<[7]%
\>[7]{}\Varid{get}\mathrel{\hstretch{0.7}{>\!\!>\!\!=}}\lambda \Varid{st}\to \Varid{guard}\;(\Varid{ok}\;(\Varid{st}\mathbin{\oplus}\Varid{x}))\mathbin{\hstretch{0.7}{>\!\!>}}\Varid{put}\;(\Varid{st}\mathbin{\oplus}\Varid{x})\mathbin{\hstretch{0.7}{>\!\!>}}{}\<[E]%
\\
\>[B]{}\hsindent{7}{}\<[7]%
\>[7]{}(\Varid{x}\mathbin{:})\mathrel{\raisebox{0.5\depth}{\scaleobj{0.5}{\langle}} \scaleobj{0.8}{\$} \raisebox{0.5\depth}{\scaleobj{0.5}{\rangle}}}\Varid{foldr}\;(\odot)\;(\Varid{return}\;[\mskip1.5mu \mskip1.5mu])\;\Varid{xs}{}\<[E]%
\\
\>[B]{}\mathbin{=}{}\<[9]%
\>[9]{}\mbox{\commentbegin  definition of \ensuremath{(\odot)}  \commentend}{}\<[E]%
\\
\>[B]{}\hsindent{7}{}\<[7]%
\>[7]{}\Varid{foldr}\;(\odot)\;(\Varid{return}\;[\mskip1.5mu \mskip1.5mu])\;(\Varid{x}\mathbin{:}\Varid{xs})~~.{}\<[E]%
\ColumnHook
\end{hscode}\resethooks
\end{proof}
This proof is instructive due to extensive use of commutativity.

In summary, we now have this corollary performing \ensuremath{\Varid{filt}\;(\Varid{all}\;\Varid{ok}\mathbin{\cdot}\Varid{scanl}_{+}\;(\oplus)\;\Varid{st})} using a non-deterministic and stateful foldr:
\begin{corollary}\label{thm:filt-scanlp-foldr} Let \ensuremath{(\odot)} be defined as in Theorem \ref{lma:foldr-guard-fusion}. If state and non-determinism commute, we have:
\begin{hscode}\SaveRestoreHook
\column{B}{@{}>{\hspre}l<{\hspost}@{}}%
\column{4}{@{}>{\hspre}l<{\hspost}@{}}%
\column{E}{@{}>{\hspre}l<{\hspost}@{}}%
\>[B]{}\Varid{filt}\;(\Varid{all}\;\Varid{ok}\mathbin{\cdot}\Varid{scanl}_{+}\;(\oplus)\;\Varid{st})\;\Varid{xs}\mathbin{=}{}\<[E]%
\\
\>[B]{}\hsindent{4}{}\<[4]%
\>[4]{}\Varid{protect}\;(\Varid{put}\;\Varid{st}\mathbin{\hstretch{0.7}{>\!\!>}}\Varid{foldr}\;(\odot)\;(\Varid{return}\;[\mskip1.5mu \mskip1.5mu])\;\Varid{xs})~~.{}\<[E]%
\ColumnHook
\end{hscode}\resethooks
\end{corollary}

\section{Monadic Hylomorphism}
\label{sec:nd-state-local}

To recap what we have done,
we started with a specification of the form
\ensuremath{\Varid{unfoldM}\;\Varid{p}\;\Varid{f}\;\Varid{z}\mathrel{\hstretch{0.7}{>\!\!>\!\!=}}\Varid{filt}\;(\Varid{all}\;\Varid{ok}\mathbin{\cdot}\Varid{scanl}_{+}\;(\oplus)\;\Varid{st})}, where
\ensuremath{\Varid{f}\mathbin{::}\Conid{MonadPlus}\;\Varid{m}\Rightarrow \Varid{b}\to \Varid{m}\;(\Varid{a},\Varid{b})}, and have shown that
\begin{hscode}\SaveRestoreHook
\column{B}{@{}>{\hspre}l<{\hspost}@{}}%
\column{7}{@{}>{\hspre}l<{\hspost}@{}}%
\column{8}{@{}>{\hspre}l<{\hspost}@{}}%
\column{9}{@{}>{\hspre}l<{\hspost}@{}}%
\column{E}{@{}>{\hspre}l<{\hspost}@{}}%
\>[7]{}\Varid{unfoldM}\;\Varid{p}\;\Varid{f}\;\Varid{z}\mathrel{\hstretch{0.7}{>\!\!>\!\!=}}\Varid{filt}\;(\Varid{all}\;\Varid{ok}\mathbin{\cdot}\Varid{scanl}_{+}\;(\oplus)\;\Varid{st}){}\<[E]%
\\
\>[B]{}\mathbin{=}{}\<[9]%
\>[9]{}\mbox{\commentbegin  Corollary \ref{thm:filt-scanlp-foldr}, with \ensuremath{(\odot)} defined as in Theorem \ref{lma:foldr-guard-fusion}  \commentend}{}\<[E]%
\\
\>[B]{}\hsindent{7}{}\<[7]%
\>[7]{}\Varid{unfoldM}\;\Varid{p}\;\Varid{f}\;\Varid{z}\mathrel{\hstretch{0.7}{>\!\!>\!\!=}}\lambda \Varid{xs}\to \Varid{protect}\;(\Varid{put}\;\Varid{st}\mathbin{\hstretch{0.7}{>\!\!>}}\Varid{foldr}\;(\odot)\;(\Varid{return}\;[\mskip1.5mu \mskip1.5mu])\;\Varid{xs}){}\<[E]%
\\
\>[B]{}\mathbin{=}{}\<[8]%
\>[8]{}\mbox{\commentbegin  Theorem~\ref{thm:nondet-commute}: nondeterminism commutes with state  \commentend}{}\<[E]%
\\
\>[B]{}\hsindent{7}{}\<[7]%
\>[7]{}\Varid{protect}\;(\Varid{put}\;\Varid{st}\mathbin{\hstretch{0.7}{>\!\!>}}\Varid{unfoldM}\;\Varid{p}\;\Varid{f}\;\Varid{z}\mathrel{\hstretch{0.7}{>\!\!>\!\!=}}\Varid{foldr}\;(\odot)\;(\Varid{return}\;[\mskip1.5mu \mskip1.5mu]))~~.{}\<[E]%
\ColumnHook
\end{hscode}\resethooks
The final task is to fuse \ensuremath{\Varid{unfoldM}\;\Varid{p}\;\Varid{f}} with \ensuremath{\Varid{foldr}\;(\odot)\;(\Varid{return}\;[\mskip1.5mu \mskip1.5mu])}.

\subsection{Monadic Hylo-Fusion}

In a pure setting, it is known that, provided that the unfolding phase terminates, \ensuremath{\Varid{foldr}\;(\otimes)\;\Varid{e}\mathbin{\cdot}\Varid{unfoldr}\;\Varid{p}\;\Varid{f}} is the unique solution of \ensuremath{\Varid{hylo}} in the equation below~\cite{Hinze:15:Conjugate}:
\begin{hscode}\SaveRestoreHook
\column{B}{@{}>{\hspre}l<{\hspost}@{}}%
\column{9}{@{}>{\hspre}l<{\hspost}@{}}%
\column{22}{@{}>{\hspre}l<{\hspost}@{}}%
\column{E}{@{}>{\hspre}l<{\hspost}@{}}%
\>[B]{}\Varid{hylo}\;\Varid{y}{}\<[9]%
\>[9]{}\mid \Varid{p}\;\Varid{y}{}\<[22]%
\>[22]{}\mathrel{=}\Varid{e}{}\<[E]%
\\
\>[9]{}\mid \Varid{otherwise}{}\<[22]%
\>[22]{}\mathrel{=}\mathbf{let}\;\Varid{f}\;\Varid{y}\mathrel{=}(\Varid{x},\Varid{z})\;\mathbf{in}\;\Varid{x}\mathbin{\otimes}\Varid{hylo}\;\Varid{z}~~.{}\<[E]%
\ColumnHook
\end{hscode}\resethooks
Hylomorphisms with monadic folds and unfolds are a bit tricky.
Pardo \shortcite{Pardo:01:Fusion} discussed hylomorphism for regular base functors, where the unfolding phase is monadic while the folding phase is pure.
As for the case when both phases are monadic, he noted ``the drawback ... is that they cannot be always transformed into a single function that avoids the construction of the intermediate data structure.''

For our purpose, we focus our attention on lists, and have a theorem fusing the monadic unfolding and folding phases under a side condition.
Given \ensuremath{(\otimes)\mathbin{::}\Varid{b}\to \Varid{m}\;\Varid{c}\to \Varid{m}\;\Varid{c}}, \ensuremath{\Varid{e}\mathbin{::}\Varid{c}}, \ensuremath{\Varid{p}\mathbin{::}\Varid{a}\to \Conid{Bool}}, and \ensuremath{\Varid{f}\mathbin{::}\Varid{a}\to \Varid{m}\;(\Varid{b},\Varid{a})} (where \ensuremath{\Conid{Monad}\;\Varid{m}}), consider the expression:
\begin{hscode}\SaveRestoreHook
\column{B}{@{}>{\hspre}l<{\hspost}@{}}%
\column{3}{@{}>{\hspre}l<{\hspost}@{}}%
\column{E}{@{}>{\hspre}l<{\hspost}@{}}%
\>[3]{}\Varid{unfoldM}\;\Varid{p}\;\Varid{f}\mathrel{\hstretch{0.7}{>\!\!=\!\!\!>}}\Varid{foldr}\;(\otimes)\;(\Varid{return}\;\Varid{e})~\,\mathbin{::}~\,\Conid{Monad}\;\Varid{m}\Rightarrow \Varid{a}\to \Varid{m}\;\Varid{c}~~.{}\<[E]%
\ColumnHook
\end{hscode}\resethooks
The following theorem says that this combination of folding and unfolding can be fused into one, with some side conditions:
\begin{theorem} \label{thm:hylo-fusion}
Let \ensuremath{\Varid{m}\mathbin{::}\mathbin{*}\to \mathbin{*}} be in type class \ensuremath{\Conid{Monad}}.
For all \ensuremath{(\otimes)\mathbin{::}\Varid{a}\to \Varid{m}\;\Varid{c}\to \Varid{m}\;\Varid{c}}, \ensuremath{\Varid{e}\mathbin{::}\Varid{m}\;\Varid{c}}, \ensuremath{\Varid{p}\mathbin{::}\Varid{b}\to \Conid{Bool}}, and \ensuremath{\Varid{f}\mathbin{::}\Varid{b}\to \Varid{m}\;(\Varid{a},\Varid{c})}, we have that \ensuremath{\Varid{unfoldM}\;\Varid{p}\;\Varid{f}\mathrel{\hstretch{0.7}{>\!\!=\!\!\!>}}\Varid{foldr}\;(\otimes)\;\Varid{e}\mathrel{=}\Varid{hyloM}\;(\otimes)\;\Varid{e}\;\Varid{p}\;\Varid{f}}, defined by:
\begin{hscode}\SaveRestoreHook
\column{B}{@{}>{\hspre}l<{\hspost}@{}}%
\column{23}{@{}>{\hspre}l<{\hspost}@{}}%
\column{36}{@{}>{\hspre}l<{\hspost}@{}}%
\column{57}{@{}>{\hspre}l<{\hspost}@{}}%
\column{E}{@{}>{\hspre}l<{\hspost}@{}}%
\>[B]{}\Varid{hyloM}\;(\otimes)\;\Varid{e}\;\Varid{p}\;\Varid{f}\;\Varid{y}{}\<[23]%
\>[23]{}\mid \Varid{p}\;\Varid{y}{}\<[36]%
\>[36]{}\mathrel{=}\Varid{e}{}\<[E]%
\\
\>[23]{}\mid \Varid{otherwise}{}\<[36]%
\>[36]{}\mathrel{=}\Varid{f}\;\Varid{y}\mathrel{\hstretch{0.7}{>\!\!>\!\!=}}\lambda (\Varid{x},\Varid{z})\to {}\<[57]%
\>[57]{}\Varid{x}\mathbin{\otimes}\Varid{hyloM}\;(\otimes)\;\Varid{e}\;\Varid{p}\;\Varid{f}\;\Varid{z}~~,{}\<[E]%
\ColumnHook
\end{hscode}\resethooks
if the relation \ensuremath{(\neg \mathbin{\cdot}\Varid{p})\mathbin{?}\mathbin{\cdot}\Varid{snd}\mathbin{\cdot}(\mathbin{\hstretch{0.7}{=\!\!<\!\!<}})\mathbin{\cdot}\Varid{f}} is well-founded (see the note below) and, for all \ensuremath{\Varid{k}}, we have
\begin{align}
 \ensuremath{\Varid{n}\mathrel{\hstretch{0.7}{>\!\!>\!\!=}}((\Varid{x}\mathbin{\otimes})\mathbin{\cdot}\Varid{k})\mathbin{=}\Varid{x}\mathbin{\otimes}(\Varid{n}\mathrel{\hstretch{0.7}{>\!\!>\!\!=}}\Varid{k})}  \mbox{~~,} \label{eq:hylo-fusion-prem}
\end{align}
where \ensuremath{\Varid{n}} abbreviates \ensuremath{\Varid{unfoldM}\;\Varid{p}\;\Varid{f}\;\Varid{z}}.
\end{theorem}
The ``well-foundedness'' condition essentially says that \ensuremath{\Varid{f}} eventually terminates --- details to be explained after the proof of this theorem.
Condition \eqref{eq:hylo-fusion-prem} may look quite restrictive.
In most cases the author have seen, however, we can actually prove that \eqref{eq:hylo-fusion-prem} holds for an entire class of \ensuremath{\Varid{n}} that includes \ensuremath{\Varid{unfoldM}\;\Varid{p}\;\Varid{f}\;\Varid{z}}.
In the application of this report, for example, the only effect of \ensuremath{\Varid{unfoldM}\;\Varid{p}\;\Varid{f}\;\Varid{z}} is non-determinism, and we will prove in Lemma~\ref{lma:fusion-condition-odot} that \eqref{eq:hylo-fusion-prem} holds for all \ensuremath{\Varid{n}} whose only effect is non-determinism, for the particular operator we use in the \ensuremath{\Varid{n}}-queens problem.

We prove Theorem~\ref{thm:hylo-fusion} below.
\begin{proof}
We start with showing that \ensuremath{\Varid{unfoldM}\;\Varid{p}\;\Varid{f}\mathrel{\hstretch{0.7}{>\!\!=\!\!\!>}}\Varid{foldr}\;(\otimes)\;\Varid{e}} is a fixed-point of the recursive equations of \ensuremath{\Varid{hyloM}}. When \ensuremath{\Varid{p}\;\Varid{y}} holds, it is immediate that
\begin{hscode}\SaveRestoreHook
\column{B}{@{}>{\hspre}l<{\hspost}@{}}%
\column{3}{@{}>{\hspre}l<{\hspost}@{}}%
\column{E}{@{}>{\hspre}l<{\hspost}@{}}%
\>[3]{}\Varid{return}\;[\mskip1.5mu \mskip1.5mu]\mathrel{\hstretch{0.7}{>\!\!>\!\!=}}\Varid{foldr}\;(\otimes)\;\Varid{e}~\mathrel{=}~\Varid{e}~~.{}\<[E]%
\ColumnHook
\end{hscode}\resethooks
When \ensuremath{\neg \;(\Varid{p}\;\Varid{y})}, we reason:
\begin{hscode}\SaveRestoreHook
\column{B}{@{}>{\hspre}l<{\hspost}@{}}%
\column{5}{@{}>{\hspre}l<{\hspost}@{}}%
\column{9}{@{}>{\hspre}l<{\hspost}@{}}%
\column{55}{@{}>{\hspre}l<{\hspost}@{}}%
\column{E}{@{}>{\hspre}l<{\hspost}@{}}%
\>[5]{}\Varid{unfoldM}\;\Varid{p}\;\Varid{f}\;\Varid{y}\mathrel{\hstretch{0.7}{>\!\!>\!\!=}}\Varid{foldr}\;(\otimes)\;\Varid{e}{}\<[E]%
\\
\>[B]{}\mathbin{=}{}\<[9]%
\>[9]{}\mbox{\commentbegin  definition of \ensuremath{\Varid{unfoldM}}, \ensuremath{\neg \;(\Varid{p}\;\Varid{y})}  \commentend}{}\<[E]%
\\
\>[B]{}\hsindent{5}{}\<[5]%
\>[5]{}(\Varid{f}\;\Varid{y}\mathrel{\hstretch{0.7}{>\!\!>\!\!=}}(\lambda (\Varid{x},\Varid{z})\to (\Varid{x}\mathbin{:})\mathrel{\raisebox{0.5\depth}{\scaleobj{0.5}{\langle}} \scaleobj{0.8}{\$} \raisebox{0.5\depth}{\scaleobj{0.5}{\rangle}}}\Varid{unfoldM}\;\Varid{p}\;\Varid{f}\;\Varid{z}))\mathrel{\hstretch{0.7}{>\!\!>\!\!=}}\Varid{foldr}\;(\otimes)\;\Varid{e}{}\<[E]%
\\
\>[B]{}\mathbin{=}{}\<[9]%
\>[9]{}\mbox{\commentbegin  monad law  and \ensuremath{\Varid{foldr}}  \commentend}{}\<[E]%
\\
\>[B]{}\hsindent{5}{}\<[5]%
\>[5]{}\Varid{f}\;\Varid{y}\mathrel{\hstretch{0.7}{>\!\!>\!\!=}}(\lambda (\Varid{x},\Varid{z})\to \Varid{unfoldM}\;\Varid{p}\;\Varid{f}\;\Varid{z}\mathrel{\hstretch{0.7}{>\!\!>\!\!=}}\lambda \Varid{xs}\to \Varid{x}\mathbin{\otimes}\Varid{foldr}\;(\otimes)\;\Varid{e}\;\Varid{xs}){}\<[E]%
\\
\>[B]{}\mathbin{=}{}\<[9]%
\>[9]{}\mbox{\commentbegin  since \ensuremath{\Varid{n}\mathrel{\hstretch{0.7}{>\!\!>\!\!=}}((\Varid{x}\mathbin{\otimes})\mathbin{\cdot}\Varid{k})\mathbin{=}\Varid{x}\mathbin{\otimes}(\Varid{n}\mathrel{\hstretch{0.7}{>\!\!>\!\!=}}\Varid{k})} where \ensuremath{\Varid{n}\mathrel{=}\Varid{unfoldM}\;\Varid{p}\;\Varid{f}\;\Varid{z}}  \commentend}{}\<[E]%
\\
\>[B]{}\hsindent{5}{}\<[5]%
\>[5]{}\Varid{f}\;\Varid{y}\mathrel{\hstretch{0.7}{>\!\!>\!\!=}}(\lambda (\Varid{x},\Varid{z})\to \Varid{x}\mathbin{\otimes}(\Varid{unfoldM}\;\Varid{p}\;\Varid{f}\;\Varid{z}\mathrel{\hstretch{0.7}{>\!\!>\!\!=}}{}\<[55]%
\>[55]{}\Varid{foldr}\;(\otimes)\;\Varid{e}))~~.{}\<[E]%
\ColumnHook
\end{hscode}\resethooks
Now that \ensuremath{\Varid{unfoldM}\;\Varid{p}\;\Varid{f}\;\Varid{z}\mathrel{\hstretch{0.7}{>\!\!>\!\!=}}\Varid{foldr}\;(\otimes)\;\Varid{e}} is a fixed-point, we may conclude that it equals \ensuremath{\Varid{hyloM}\;(\otimes)\;\Varid{e}\;\Varid{p}\;\Varid{f}} if the latter has a unique fixed-point,
which is guaranteed by the well-foundedness condition.
See the note below.
\end{proof}

\paragraph{Note} Let \ensuremath{\Varid{q}} be a predicate, \ensuremath{\Varid{q}\mathbin{?}} is a relation defined by \ensuremath{\{\mskip1.5mu (\Varid{x},\Varid{x})\mid\Varid{q}\;\Varid{x}\mskip1.5mu\}}. The parameter \ensuremath{\Varid{y}} in \ensuremath{\Varid{unfoldM}} is called the {\em seed} used to generate the list. The relation \ensuremath{(\neg \mathbin{\cdot}\Varid{p})\mathbin{?}\mathbin{\cdot}\Varid{snd}\mathbin{\cdot}(\mathbin{\hstretch{0.7}{=\!\!<\!\!<}})\mathbin{\cdot}\Varid{f}} maps one seed to the next seed (where \ensuremath{(\mathbin{\hstretch{0.7}{=\!\!<\!\!<}})} is \ensuremath{(\mathrel{\hstretch{0.7}{>\!\!>\!\!=}})} written reversed). If it is {\em well-founded}, intuitively speaking, the seed generation cannot go on forever and \ensuremath{\Varid{p}} will eventually hold. It is known that inductive types (those can be folded) and coinductive types (those can be unfolded) do not coincide in {\sf SET}. To allow a fold to be composed after an unfold, typically one moves to a semantics based on complete partial orders. However, it was shown~\cite{DoornbosBackhouse:95:Induction} that, in {\sf Rel}, when the relation generating seeds is well-founded, hylo-equations do have unique solutions. One may thus stay within a set-theoretic semantics. Such an approach is recently explored again~\cite{Hinze:15:Conjugate}. ({\em End of Note})

\vspace{1em}
Theorem \ref{thm:hylo-fusion} does not rely on the \emph{local state laws} \eqref{eq:mplus-bind-dist} and \eqref{eq:mzero-bind-zero}, and does not put restriction on \ensuremath{\epsilon}.
To apply the theorem to our particular case, we have to show that its preconditions hold for our particular \ensuremath{(\odot)} ---
for that we will need \eqref{eq:mzero-bind-zero} and perhaps also \eqref{eq:mplus-bind-dist}. In the lemma below we slightly generalise \ensuremath{(\odot)} in Theorem \ref{lma:foldr-guard-fusion}:
\begin{lemma}\label{lma:fusion-condition-odot}
Assume that \eqref{eq:mzero-bind-zero} holds.
Given \ensuremath{\Varid{p}\mathbin{::}\Varid{a}\to \Varid{s}\to \Conid{Bool}}, \ensuremath{\Varid{next}\mathbin{::}\Varid{a}\to \Varid{s}\to \Varid{s}}, and \ensuremath{\Varid{res}\mathbin{::}\Varid{a}\to \Varid{b}\to \Varid{b}}, define \ensuremath{(\odot)} as below:
\begin{hscode}\SaveRestoreHook
\column{B}{@{}>{\hspre}l<{\hspost}@{}}%
\column{3}{@{}>{\hspre}l<{\hspost}@{}}%
\column{17}{@{}>{\hspre}l<{\hspost}@{}}%
\column{E}{@{}>{\hspre}l<{\hspost}@{}}%
\>[3]{}(\odot)\mathbin{::}(\Conid{MonadPlus}\;\Varid{m},\Conid{MonadState}\;\Varid{s}\;\Varid{m})\Rightarrow \Varid{a}\to \Varid{m}\;\Varid{b}\to \Varid{m}\;\Varid{b}{}\<[E]%
\\
\>[3]{}\Varid{x}\mathbin{\odot}\Varid{m}\mathrel{=}{}\<[17]%
\>[17]{}\Varid{get}\mathrel{\hstretch{0.7}{>\!\!>\!\!=}}\lambda \Varid{st}\to \Varid{guard}\;(\Varid{p}\;\Varid{x}\;\Varid{st})\mathbin{\hstretch{0.7}{>\!\!>}}{}\<[E]%
\\
\>[17]{}\Varid{put}\;(\Varid{next}\;\Varid{x}\;\Varid{st})\mathbin{\hstretch{0.7}{>\!\!>}}(\Varid{res}\;\Varid{x}\mathrel{\raisebox{0.5\depth}{\scaleobj{0.5}{\langle}} \scaleobj{0.8}{\$} \raisebox{0.5\depth}{\scaleobj{0.5}{\rangle}}}\Varid{m})~~.{}\<[E]%
\ColumnHook
\end{hscode}\resethooks
We have \ensuremath{\Varid{n}\mathrel{\hstretch{0.7}{>\!\!>\!\!=}}((\Varid{x}\mathbin{\odot})\mathbin{\cdot}\Varid{k})\mathbin{=}\Varid{x}\mathbin{\odot}(\Varid{n}\mathrel{\hstretch{0.7}{>\!\!>\!\!=}}\Varid{k})}, if \ensuremath{\Varid{n}} commutes with state.
\end{lemma}
\begin{proof} We reason:
\begin{hscode}\SaveRestoreHook
\column{B}{@{}>{\hspre}l<{\hspost}@{}}%
\column{7}{@{}>{\hspre}l<{\hspost}@{}}%
\column{8}{@{}>{\hspre}l<{\hspost}@{}}%
\column{9}{@{}>{\hspre}l<{\hspost}@{}}%
\column{E}{@{}>{\hspre}l<{\hspost}@{}}%
\>[7]{}\Varid{n}\mathrel{\hstretch{0.7}{>\!\!>\!\!=}}((\Varid{x}\mathbin{\odot})\mathbin{\cdot}\Varid{k}){}\<[E]%
\\
\>[B]{}\mathbin{=}{}\<[7]%
\>[7]{}\Varid{n}\mathrel{\hstretch{0.7}{>\!\!>\!\!=}}\lambda \Varid{y}\to \Varid{x}\mathbin{\odot}\Varid{k}\;\Varid{y}{}\<[E]%
\\
\>[B]{}\mathbin{=}{}\<[9]%
\>[9]{}\mbox{\commentbegin  definition of \ensuremath{(\odot)}  \commentend}{}\<[E]%
\\
\>[B]{}\hsindent{7}{}\<[7]%
\>[7]{}\Varid{n}\mathrel{\hstretch{0.7}{>\!\!>\!\!=}}\lambda \Varid{y}\to \Varid{get}\mathrel{\hstretch{0.7}{>\!\!>\!\!=}}\lambda \Varid{st}\to {}\<[E]%
\\
\>[B]{}\hsindent{7}{}\<[7]%
\>[7]{}\Varid{guard}\;(\Varid{p}\;\Varid{x}\;\Varid{st})\mathbin{\hstretch{0.7}{>\!\!>}}\Varid{put}\;(\Varid{next}\;\Varid{x}\;\Varid{st})\mathbin{\hstretch{0.7}{>\!\!>}}(\Varid{res}\;\Varid{x}\mathrel{\raisebox{0.5\depth}{\scaleobj{0.5}{\langle}} \scaleobj{0.8}{\$} \raisebox{0.5\depth}{\scaleobj{0.5}{\rangle}}}\Varid{k}\;\Varid{y}){}\<[E]%
\\
\>[B]{}\mathbin{=}{}\<[9]%
\>[9]{}\mbox{\commentbegin  \ensuremath{\Varid{n}} commutes with state  \commentend}{}\<[E]%
\\
\>[B]{}\hsindent{7}{}\<[7]%
\>[7]{}\Varid{get}\mathrel{\hstretch{0.7}{>\!\!>\!\!=}}\lambda \Varid{st}\to \Varid{n}\mathrel{\hstretch{0.7}{>\!\!>\!\!=}}\lambda \Varid{y}\to {}\<[E]%
\\
\>[B]{}\hsindent{7}{}\<[7]%
\>[7]{}\Varid{guard}\;(\Varid{p}\;\Varid{x}\;\Varid{st})\mathbin{\hstretch{0.7}{>\!\!>}}\Varid{put}\;(\Varid{next}\;\Varid{x}\;\Varid{st})\mathbin{\hstretch{0.7}{>\!\!>}}(\Varid{res}\;\Varid{x}\mathrel{\raisebox{0.5\depth}{\scaleobj{0.5}{\langle}} \scaleobj{0.8}{\$} \raisebox{0.5\depth}{\scaleobj{0.5}{\rangle}}}(\Varid{k}\;\Varid{y})){}\<[E]%
\\
\>[B]{}\mathbin{=}{}\<[8]%
\>[8]{}\mbox{\commentbegin  by \eqref{eq:guard-commute}, since \eqref{eq:mzero-bind-zero} holds  \commentend}{}\<[E]%
\\
\>[B]{}\hsindent{7}{}\<[7]%
\>[7]{}\Varid{get}\mathrel{\hstretch{0.7}{>\!\!>\!\!=}}\lambda \Varid{st}\to \Varid{guard}\;(\Varid{p}\;\Varid{x}\;\Varid{st})\mathbin{\hstretch{0.7}{>\!\!>}}{}\<[E]%
\\
\>[B]{}\hsindent{7}{}\<[7]%
\>[7]{}\Varid{n}\mathrel{\hstretch{0.7}{>\!\!>\!\!=}}\lambda \Varid{y}\to \Varid{put}\;(\Varid{next}\;\Varid{x}\;\Varid{st})\mathbin{\hstretch{0.7}{>\!\!>}}(\Varid{res}\;\Varid{x}\mathrel{\raisebox{0.5\depth}{\scaleobj{0.5}{\langle}} \scaleobj{0.8}{\$} \raisebox{0.5\depth}{\scaleobj{0.5}{\rangle}}}(\Varid{k}\;\Varid{y})){}\<[E]%
\\
\>[B]{}\mathbin{=}{}\<[8]%
\>[8]{}\mbox{\commentbegin  \ensuremath{\Varid{n}} commutes with state   \commentend}{}\<[E]%
\\
\>[B]{}\hsindent{7}{}\<[7]%
\>[7]{}\Varid{get}\mathrel{\hstretch{0.7}{>\!\!>\!\!=}}\lambda \Varid{st}\to \Varid{guard}\;(\Varid{p}\;\Varid{x}\;\Varid{st})\mathbin{\hstretch{0.7}{>\!\!>}}{}\<[E]%
\\
\>[B]{}\hsindent{7}{}\<[7]%
\>[7]{}\Varid{put}\;(\Varid{next}\;\Varid{x}\;\Varid{st})\mathbin{\hstretch{0.7}{>\!\!>}}\Varid{n}\mathrel{\hstretch{0.7}{>\!\!>\!\!=}}\lambda \Varid{y}\to (\Varid{res}\;\Varid{x}\mathrel{\raisebox{0.5\depth}{\scaleobj{0.5}{\langle}} \scaleobj{0.8}{\$} \raisebox{0.5\depth}{\scaleobj{0.5}{\rangle}}}(\Varid{k}\;\Varid{y})){}\<[E]%
\\
\>[B]{}\mathbin{=}{}\<[9]%
\>[9]{}\mbox{\commentbegin  properties of \ensuremath{(\mathrel{\raisebox{0.5\depth}{\scaleobj{0.5}{\langle}} \scaleobj{0.8}{\$} \raisebox{0.5\depth}{\scaleobj{0.5}{\rangle}}})}   \commentend}{}\<[E]%
\\
\>[B]{}\hsindent{7}{}\<[7]%
\>[7]{}\Varid{get}\mathrel{\hstretch{0.7}{>\!\!>\!\!=}}\lambda \Varid{st}\to \Varid{guard}\;(\Varid{p}\;\Varid{x}\;\Varid{st})\mathbin{\hstretch{0.7}{>\!\!>}}{}\<[E]%
\\
\>[B]{}\hsindent{7}{}\<[7]%
\>[7]{}\Varid{put}\;(\Varid{next}\;\Varid{x}\;\Varid{st})\mathbin{\hstretch{0.7}{>\!\!>}}(\Varid{res}\;\Varid{x}\mathrel{\raisebox{0.5\depth}{\scaleobj{0.5}{\langle}} \scaleobj{0.8}{\$} \raisebox{0.5\depth}{\scaleobj{0.5}{\rangle}}}\Varid{n}\mathrel{\hstretch{0.7}{>\!\!>\!\!=}}\Varid{k}){}\<[E]%
\\
\>[B]{}\mathbin{=}{}\<[9]%
\>[9]{}\mbox{\commentbegin  definition of \ensuremath{(\odot)}  \commentend}{}\<[E]%
\\
\>[B]{}\hsindent{7}{}\<[7]%
\>[7]{}\Varid{x}\mathbin{\odot}(\Varid{n}\mathrel{\hstretch{0.7}{>\!\!>\!\!=}}\Varid{k})~~.{}\<[E]%
\ColumnHook
\end{hscode}\resethooks
\end{proof}

\subsection{Summary, and Solving \ensuremath{\Varid{n}}-Queens}
\label{sec:solve-n-queens}

To conclude our derivation, a problem formulated as \ensuremath{\Varid{unfoldM}\;\Varid{p}\;\Varid{f}\;\Varid{z}\mathrel{\hstretch{0.7}{>\!\!>\!\!=}}\Varid{filt}\;(\Varid{all}\;\Varid{ok}\mathbin{\cdot}\Varid{scanl}_{+}\;(\oplus)\;\Varid{st})} can be solved by a hylomorphism. Define:
\begin{hscode}\SaveRestoreHook
\column{B}{@{}>{\hspre}l<{\hspost}@{}}%
\column{3}{@{}>{\hspre}l<{\hspost}@{}}%
\column{23}{@{}>{\hspre}l<{\hspost}@{}}%
\column{E}{@{}>{\hspre}l<{\hspost}@{}}%
\>[B]{}\Varid{solve}\mathbin{::}(\Conid{MonadState}\;\Varid{s}\;\Varid{m},\Conid{MonadPlus}\;\Varid{m})\Rightarrow {}\<[E]%
\\
\>[B]{}\hsindent{3}{}\<[3]%
\>[3]{}(\Varid{b}\to \Conid{Bool})\to (\Varid{b}\to \Varid{m}\;(\Varid{a},\Varid{b}))\to (\Varid{s}\to \Conid{Bool})\to (\Varid{s}\to \Varid{a}\to \Varid{s})\to \Varid{s}\to \Varid{b}\to \Varid{m}\;[\mskip1.5mu \Varid{a}\mskip1.5mu]{}\<[E]%
\\
\>[B]{}\Varid{solve}\;\Varid{p}\;\Varid{f}\;\Varid{ok}\;(\oplus)\;\Varid{st}\;\Varid{z}\mathrel{=}\Varid{protect}\;(\Varid{put}\;\Varid{st}\mathbin{\hstretch{0.7}{>\!\!>}}\Varid{hyloM}\;(\odot)\;(\Varid{return}\;[\mskip1.5mu \mskip1.5mu])\;\Varid{p}\;\Varid{f}\;\Varid{z})~~,{}\<[E]%
\\
\>[B]{}\hsindent{3}{}\<[3]%
\>[3]{}\mathbf{where}\;\Varid{x}\mathbin{\odot}\Varid{m}\mathrel{=}{}\<[23]%
\>[23]{}\Varid{get}\mathrel{\hstretch{0.7}{>\!\!>\!\!=}}\lambda \Varid{st}\to \Varid{guard}\;(\Varid{ok}\;(\Varid{st}\mathbin{\oplus}\Varid{x}))\mathbin{\hstretch{0.7}{>\!\!>}}{}\<[E]%
\\
\>[23]{}\Varid{put}\;(\Varid{st}\mathbin{\oplus}\Varid{x})\mathbin{\hstretch{0.7}{>\!\!>}}((\Varid{x}\mathbin{:})\mathrel{\raisebox{0.5\depth}{\scaleobj{0.5}{\langle}} \scaleobj{0.8}{\$} \raisebox{0.5\depth}{\scaleobj{0.5}{\rangle}}}\Varid{m})~~.{}\<[E]%
\ColumnHook
\end{hscode}\resethooks
\begin{corollary} \label{cor:unfold-filt-scanl-local}
Given \ensuremath{\Varid{p}\mathbin{::}\Varid{b}\to \Conid{Bool}}, \ensuremath{\Varid{f}\mathbin{::}(\Conid{MonadPlus}\;\Varid{m},\Conid{MonadState}\;\Varid{s}\;\Varid{m})\Rightarrow \Varid{b}\to \Varid{m}\;(\Varid{a},\Varid{b})}, \ensuremath{\Varid{z}\mathbin{::}\Varid{b}}, \ensuremath{\Varid{ok}\mathbin{::}\Varid{s}\to \Conid{Bool}}, \ensuremath{(\oplus)\mathbin{::}\Varid{s}\to \Varid{a}\to \Varid{s}}, \ensuremath{\Varid{st}\mathbin{::}\Varid{s}},
If the relation \ensuremath{(\neg \mathbin{\cdot}\Varid{p})\mathbin{?}\mathbin{\cdot}\Varid{snd}\mathbin{\cdot}(\mathbin{\hstretch{0.7}{=\!\!<\!\!<}})\mathbin{\cdot}\Varid{f}} is well-founded,
the local state laws hold in addition to the other laws,
and \ensuremath{\Varid{unfoldM}\;\Varid{p}\;\Varid{f}\;\Varid{z}} commutes with state, we have
\begin{hscode}\SaveRestoreHook
\column{B}{@{}>{\hspre}l<{\hspost}@{}}%
\column{5}{@{}>{\hspre}l<{\hspost}@{}}%
\column{E}{@{}>{\hspre}l<{\hspost}@{}}%
\>[B]{}\Varid{unfoldM}\;\Varid{p}\;\Varid{f}\;\Varid{z}\mathrel{\hstretch{0.7}{>\!\!>\!\!=}}\Varid{filt}\;(\Varid{all}\;\Varid{ok}\mathbin{\cdot}\Varid{scanl}_{+}\;(\oplus)\;\Varid{st})\mathbin{=}{}\<[E]%
\\
\>[B]{}\hsindent{5}{}\<[5]%
\>[5]{}\Varid{solve}\;\Varid{p}\;\Varid{f}\;\Varid{ok}\;(\oplus)\;\Varid{st}\;\Varid{z}~~.{}\<[E]%
\ColumnHook
\end{hscode}\resethooks
\end{corollary}

\paragraph{\ensuremath{\Varid{n}}-Queens Solved}
Recall that
\begin{hscode}\SaveRestoreHook
\column{B}{@{}>{\hspre}l<{\hspost}@{}}%
\column{11}{@{}>{\hspre}l<{\hspost}@{}}%
\column{E}{@{}>{\hspre}l<{\hspost}@{}}%
\>[B]{}\Varid{queens}\;\Varid{n}{}\<[11]%
\>[11]{}\mathrel{=}\Varid{perm}\;[\mskip1.5mu \mathrm{0}\mathinner{\ldotp\ldotp}\Varid{n}\mathbin{-}\mathrm{1}\mskip1.5mu]\mathrel{\hstretch{0.7}{>\!\!>\!\!=}}\Varid{filt}\;\Varid{safe}{}\<[E]%
\\
\>[11]{}\mathrel{=}\Varid{unfoldM}\;\Varid{null}\;\Varid{select}\;[\mskip1.5mu \mathrm{0}\mathinner{\ldotp\ldotp}\Varid{n}\mathbin{-}\mathrm{1}\mskip1.5mu]\mathrel{\hstretch{0.7}{>\!\!>\!\!=}}\Varid{filt}\;(\Varid{all}\;\Varid{ok}\mathbin{\cdot}\Varid{scanl}_{+}\;(\oplus)\;(\mathrm{0},[\mskip1.5mu \mskip1.5mu],[\mskip1.5mu \mskip1.5mu]))~~,{}\<[E]%
\ColumnHook
\end{hscode}\resethooks
where the auxiliary functions \ensuremath{\Varid{select}}, \ensuremath{\Varid{ok}}, \ensuremath{(\oplus)} are defined in Section \ref{sec:queens}.
The function \ensuremath{\Varid{select}} cannot be applied forever since the length of the given list decreases after each call,
and \ensuremath{\Varid{perm}}, using only non-determinism, commutes with state.
Therefore, Corollary \ref{cor:unfold-filt-scanl-local} applies, and we have \ensuremath{\Varid{queens}\;\Varid{n}\mathrel{=}\Varid{solve}\;\Varid{null}\;\Varid{select}\;\Varid{ok}\;(\oplus)\;(\mathrm{0},[\mskip1.5mu \mskip1.5mu],[\mskip1.5mu \mskip1.5mu])\;[\mskip1.5mu \mathrm{0}\mathinner{\ldotp\ldotp}\Varid{n}\mathbin{-}\mathrm{1}\mskip1.5mu]}.
Expanding the definitions we get:
\begin{hscode}\SaveRestoreHook
\column{B}{@{}>{\hspre}l<{\hspost}@{}}%
\column{3}{@{}>{\hspre}l<{\hspost}@{}}%
\column{10}{@{}>{\hspre}l<{\hspost}@{}}%
\column{16}{@{}>{\hspre}c<{\hspost}@{}}%
\column{16E}{@{}l@{}}%
\column{19}{@{}>{\hspre}l<{\hspost}@{}}%
\column{E}{@{}>{\hspre}l<{\hspost}@{}}%
\>[B]{}\Varid{queens}\mathbin{::}(\Conid{MonadPlus}\;\Varid{m},\Conid{MonadState}\;(\Conid{Int},[\mskip1.5mu \Conid{Int}\mskip1.5mu],[\mskip1.5mu \Conid{Int}\mskip1.5mu])\;\Varid{m})\Rightarrow \Conid{Int}\to \Varid{m}\;[\mskip1.5mu \Conid{Int}\mskip1.5mu]{}\<[E]%
\\
\>[B]{}\Varid{queens}\;\Varid{n}\mathrel{=}\Varid{protect}\;(\Varid{put}\;(\mathrm{0},[\mskip1.5mu \mskip1.5mu],[\mskip1.5mu \mskip1.5mu])\mathbin{\hstretch{0.7}{>\!\!>}}\Varid{queensBody}\;[\mskip1.5mu \mathrm{0}\mathinner{\ldotp\ldotp}\Varid{n}\mathbin{-}\mathrm{1}\mskip1.5mu])~~,{}\<[E]%
\\[\blanklineskip]%
\>[B]{}\Varid{queensBody}\mathbin{::}(\Conid{MonadPlus}\;\Varid{m},\Conid{MonadState}\;(\Conid{Int},[\mskip1.5mu \Conid{Int}\mskip1.5mu],[\mskip1.5mu \Conid{Int}\mskip1.5mu])\;\Varid{m})\Rightarrow [\mskip1.5mu \Conid{Int}\mskip1.5mu]\to \Varid{m}\;[\mskip1.5mu \Conid{Int}\mskip1.5mu]{}\<[E]%
\\
\>[B]{}\Varid{queensBody}\;[\mskip1.5mu \mskip1.5mu]{}\<[16]%
\>[16]{}\mathrel{=}{}\<[16E]%
\>[19]{}\Varid{return}\;[\mskip1.5mu \mskip1.5mu]{}\<[E]%
\\
\>[B]{}\Varid{queensBody}\;\Varid{xs}{}\<[16]%
\>[16]{}\mathrel{=}{}\<[16E]%
\>[19]{}\Varid{select}\;\Varid{xs}\mathrel{\hstretch{0.7}{>\!\!>\!\!=}}\lambda (\Varid{x},\Varid{ys})\to {}\<[E]%
\\
\>[19]{}\Varid{get}\mathrel{\hstretch{0.7}{>\!\!>\!\!=}}\lambda \Varid{st}\to \Varid{guard}\;(\Varid{ok}\;(\Varid{st}\mathbin{\oplus}\Varid{x}))\mathbin{\hstretch{0.7}{>\!\!>}}{}\<[E]%
\\
\>[19]{}\Varid{put}\;(\Varid{st}\mathbin{\oplus}\Varid{x})\mathbin{\hstretch{0.7}{>\!\!>}}((\Varid{x}\mathbin{:})\mathrel{\raisebox{0.5\depth}{\scaleobj{0.5}{\langle}} \scaleobj{0.8}{\$} \raisebox{0.5\depth}{\scaleobj{0.5}{\rangle}}}\Varid{queensBody}\;\Varid{ys})~~,{}\<[E]%
\\
\>[B]{}\hsindent{3}{}\<[3]%
\>[3]{}\mathbf{where}\;{}\<[10]%
\>[10]{}(\Varid{i},\Varid{us},\Varid{ds})\mathbin{\oplus}\Varid{x}\mathrel{=}(\mathrm{1}\mathbin{+}\Varid{i},(\Varid{i}\mathbin{+}\Varid{x})\mathbin{:}\Varid{us},(\Varid{i}\mathbin{-}\Varid{x})\mathbin{:}\Varid{ds}){}\<[E]%
\\
\>[10]{}\Varid{ok}\;(\anonymous ,\Varid{u}\mathbin{:}\Varid{us},\Varid{d}\mathbin{:}\Varid{ds})\mathrel{=}(\Varid{u}\notin \Varid{us})\mathrel{\wedge}(\Varid{d}\notin \Varid{ds})~~.{}\<[E]%
\ColumnHook
\end{hscode}\resethooks
This completes the derivation of our backtracking algorithm for the \ensuremath{\Varid{n}}-queens problem.
\section{Conclusions and Related Work}
\label{sec:conclusion}

This report is a case study of reasoning and derivation of monadic programs.
To study the interaction between non-determinism and state, we
construct backtracking algorithms solving problems that can be specified in the form \ensuremath{\Varid{unfoldM}\;\Varid{f}\;\Varid{p}\mathrel{\hstretch{0.7}{>\!\!=\!\!\!>}}\Varid{assert}\;(\Varid{all}\;\Varid{ok}\mathbin{\cdot}\Varid{scanl}_{+}\;(\oplus)\;\Varid{st})}.
The derivation of the backtracking algorithm works by fusing the two phases into a monadic hylomorphism.
It turns out that in derivations of programs using non-determinism and state, commutativity plays an important role.
We assume in this report that the local state laws (right-distributivity and right-zero) hold.
In this scenario we have nicer properties at hand, and commutativity holds more generally.

The local state laws imply that each non-deterministic branch has its own state.
It is cheap to implement when the state can be represented by linked data structures, such as a tuple containing lists, as in the \ensuremath{\Varid{n}}-queens example.
When the state contains blocked data, such as a large array, duplicating the state for each non-deterministic branch can be costly.
Hence there is practical need for sharing one global state among non-deterministic branches.
When a monad supports shared global state and non-determinism, commutativity of the two effects holds in limited cases. The behaviour of the monad is much less intuitive, and might be considered awkward sometimes.
In a subsequent paper~\citep{Pauwels:19:Handling},
we attempt to find out what algebraic laws we can expect and how to reason with programs when the state is global.

\citet{Affeldt:19:Hierarchy} modelled a hierarchy of monadic effects in Coq. The formalisation was applied to verify a number of equational proofs of monadic programs, including some of the proofs in an earlier version of this report. A number of errors was found and reported to the author. 

\subsubsection*{Acknowledgements} The author would like to thank Tyng-Ruey Chuang for examining a very early draft of this report; Jeremy Gibbons, who has been following the development of this research and keeping giving good advices; and Tom Schrijvers and Koen Pauwels, for nice cooperation on work following-up this report.
Thanks also go to Reynald Affeldt, David Nowak and Takafumi Saikawa for verifying and finding errors in the proofs in an earlier version of this report.
The author is solely responsible for any remaining errors, however.


\appendix


\section{Miscellaneous Proofs}
\label{sec:misc-proofs}

Proofs of \eqref{eq:guard-conj} and \eqref{eq:guard-fmap}.
\begin{proof}
Proof of \eqref{eq:guard-conj} relies only on property of \ensuremath{\mathbf{if}}
and conjunction:
\begin{hscode}\SaveRestoreHook
\column{B}{@{}>{\hspre}c<{\hspost}@{}}%
\column{BE}{@{}l@{}}%
\column{4}{@{}>{\hspre}l<{\hspost}@{}}%
\column{E}{@{}>{\hspre}l<{\hspost}@{}}%
\>[4]{}\Varid{guard}\;(\Varid{p}\mathrel{\wedge}\Varid{q}){}\<[E]%
\\
\>[B]{}\mathrel{=}{}\<[BE]%
\>[4]{}\mathbf{if}\;\Varid{p}\mathrel{\wedge}\Varid{q}\;\mathbf{then}\;\Varid{return}\;()\;\mathbf{else}\;\emptyset{}\<[E]%
\\
\>[B]{}\mathrel{=}{}\<[BE]%
\>[4]{}\mathbf{if}\;\Varid{p}\;\mathbf{then}\;(\mathbf{if}\;\Varid{q}\;\mathbf{then}\;\Varid{return}\;()\;\mathbf{else}\;\emptyset)\;\mathbf{else}\;\emptyset{}\<[E]%
\\
\>[B]{}\mathrel{=}{}\<[BE]%
\>[4]{}\mathbf{if}\;\Varid{p}\;\mathbf{then}\;\Varid{guard}\;\Varid{q}\;\mathbf{else}\;\emptyset{}\<[E]%
\\
\>[B]{}\mathrel{=}{}\<[BE]%
\>[4]{}\Varid{guard}\;\Varid{p}\mathbin{\hstretch{0.7}{>\!\!>}}\Varid{guard}\;\Varid{q}~~.{}\<[E]%
\ColumnHook
\end{hscode}\resethooks

To prove \eqref{eq:guard-fmap}, surprisingly, we need only
distributivity and {\em not} \eqref{eq:bind-mzero-zero}:
\begin{hscode}\SaveRestoreHook
\column{B}{@{}>{\hspre}c<{\hspost}@{}}%
\column{BE}{@{}l@{}}%
\column{4}{@{}>{\hspre}l<{\hspost}@{}}%
\column{5}{@{}>{\hspre}l<{\hspost}@{}}%
\column{E}{@{}>{\hspre}l<{\hspost}@{}}%
\>[4]{}\Varid{guard}\;\Varid{p}\mathbin{\hstretch{0.7}{>\!\!>}}(\Varid{f}\mathrel{\raisebox{0.5\depth}{\scaleobj{0.5}{\langle}} \scaleobj{0.8}{\$} \raisebox{0.5\depth}{\scaleobj{0.5}{\rangle}}}\Varid{m}){}\<[E]%
\\
\>[B]{}\mathrel{=}{}\<[BE]%
\>[5]{}\mbox{\commentbegin  definition of \ensuremath{\Varid{guard}}  \commentend}{}\<[E]%
\\
\>[B]{}\hsindent{4}{}\<[4]%
\>[4]{}(\mathbf{if}\;\Varid{p}\;\mathbf{then}\;\Varid{return}\;()\;\mathbf{else}\;\emptyset)\mathbin{\hstretch{0.7}{>\!\!>}}(\Varid{f}\mathrel{\raisebox{0.5\depth}{\scaleobj{0.5}{\langle}} \scaleobj{0.8}{\$} \raisebox{0.5\depth}{\scaleobj{0.5}{\rangle}}}\Varid{m}){}\<[E]%
\\
\>[B]{}\mathrel{=}{}\<[BE]%
\>[5]{}\mbox{\commentbegin  by \eqref{eq:if-distr}  \commentend}{}\<[E]%
\\
\>[B]{}\hsindent{4}{}\<[4]%
\>[4]{}\mathbf{if}\;\Varid{p}\;\mathbf{then}\;\Varid{return}\;()\mathbin{\hstretch{0.7}{>\!\!>}}(\Varid{f}\mathrel{\raisebox{0.5\depth}{\scaleobj{0.5}{\langle}} \scaleobj{0.8}{\$} \raisebox{0.5\depth}{\scaleobj{0.5}{\rangle}}}\Varid{m})\;\mathbf{else}\;\emptyset\mathbin{>>>}(\Varid{f}\mathrel{\raisebox{0.5\depth}{\scaleobj{0.5}{\langle}} \scaleobj{0.8}{\$} \raisebox{0.5\depth}{\scaleobj{0.5}{\rangle}}}\Varid{m}){}\<[E]%
\\
\>[B]{}\mathrel{=}{}\<[BE]%
\>[5]{}\mbox{\commentbegin  by \eqref{eq:monad-assoc} and \eqref{eq:monad-bind-ret}  \commentend}{}\<[E]%
\\
\>[B]{}\hsindent{4}{}\<[4]%
\>[4]{}\mathbf{if}\;\Varid{p}\;\mathbf{then}\;\Varid{f}\mathrel{\raisebox{0.5\depth}{\scaleobj{0.5}{\langle}} \scaleobj{0.8}{\$} \raisebox{0.5\depth}{\scaleobj{0.5}{\rangle}}}(\Varid{return}\;()\mathbin{\hstretch{0.7}{>\!\!>}}\Varid{m})\;\mathbf{else}\;\Varid{f}\mathrel{\raisebox{0.5\depth}{\scaleobj{0.5}{\langle}} \scaleobj{0.8}{\$} \raisebox{0.5\depth}{\scaleobj{0.5}{\rangle}}}(\emptyset\mathbin{\hstretch{0.7}{>\!\!>}}\Varid{m}){}\<[E]%
\\
\>[B]{}\mathrel{=}{}\<[BE]%
\>[5]{}\mbox{\commentbegin  by \eqref{eq:if-distr}  \commentend}{}\<[E]%
\\
\>[B]{}\hsindent{4}{}\<[4]%
\>[4]{}\Varid{f}\mathrel{\raisebox{0.5\depth}{\scaleobj{0.5}{\langle}} \scaleobj{0.8}{\$} \raisebox{0.5\depth}{\scaleobj{0.5}{\rangle}}}((\mathbf{if}\;\Varid{p}\;\mathbf{then}\;\Varid{return}\;()\;\mathbf{else}\;\emptyset)\mathbin{\hstretch{0.7}{>\!\!>}}\Varid{m}){}\<[E]%
\\
\>[B]{}\mathrel{=}{}\<[BE]%
\>[5]{}\mbox{\commentbegin  definition of \ensuremath{\Varid{guard}}  \commentend}{}\<[E]%
\\
\>[B]{}\hsindent{4}{}\<[4]%
\>[4]{}\Varid{f}\mathrel{\raisebox{0.5\depth}{\scaleobj{0.5}{\langle}} \scaleobj{0.8}{\$} \raisebox{0.5\depth}{\scaleobj{0.5}{\rangle}}}(\Varid{guard}\;\Varid{p}\mathbin{\hstretch{0.7}{>\!\!>}}\Varid{m})~~.{}\<[E]%
\ColumnHook
\end{hscode}\resethooks
\end{proof}

\noindent Proof of \eqref{eq:guard-commute}.
\begin{proof} We reason:
\begin{hscode}\SaveRestoreHook
\column{B}{@{}>{\hspre}c<{\hspost}@{}}%
\column{BE}{@{}l@{}}%
\column{4}{@{}>{\hspre}l<{\hspost}@{}}%
\column{5}{@{}>{\hspre}l<{\hspost}@{}}%
\column{9}{@{}>{\hspre}l<{\hspost}@{}}%
\column{10}{@{}>{\hspre}l<{\hspost}@{}}%
\column{16}{@{}>{\hspre}l<{\hspost}@{}}%
\column{E}{@{}>{\hspre}l<{\hspost}@{}}%
\>[4]{}\Varid{m}\mathrel{\hstretch{0.7}{>\!\!>\!\!=}}\lambda \Varid{x}\to \Varid{guard}\;\Varid{p}\mathbin{\hstretch{0.7}{>\!\!>}}\Varid{return}\;\Varid{x}{}\<[E]%
\\
\>[B]{}\mathrel{=}{}\<[BE]%
\>[5]{}\mbox{\commentbegin  definition of \ensuremath{\Varid{guard}}  \commentend}{}\<[E]%
\\
\>[B]{}\hsindent{4}{}\<[4]%
\>[4]{}\Varid{m}\mathrel{\hstretch{0.7}{>\!\!>\!\!=}}\lambda \Varid{x}\to (\mathbf{if}\;\Varid{p}\;\mathbf{then}\;\Varid{return}\;()\;\mathbf{else}\;\emptyset)\mathbin{\hstretch{0.7}{>\!\!>}}\Varid{return}\;\Varid{x}{}\<[E]%
\\
\>[B]{}\mathrel{=}{}\<[BE]%
\>[5]{}\mbox{\commentbegin  by \eqref{eq:if-distr}, with \ensuremath{\Varid{f}\;\Varid{n}\mathrel{=}\Varid{m}\mathrel{\hstretch{0.7}{>\!\!>\!\!=}}\lambda \Varid{x}\to \Varid{n}\mathbin{\hstretch{0.7}{>\!\!>}}\Varid{return}\;\Varid{x}}  \commentend}{}\<[E]%
\\
\>[B]{}\hsindent{4}{}\<[4]%
\>[4]{}\mathbf{if}\;\Varid{p}\;{}\<[10]%
\>[10]{}\mathbf{then}\;{}\<[16]%
\>[16]{}\Varid{m}\mathrel{\hstretch{0.7}{>\!\!>\!\!=}}\lambda \Varid{x}\to \Varid{return}\;()\mathbin{\hstretch{0.7}{>\!\!>}}\Varid{return}\;\Varid{x}{}\<[E]%
\\
\>[10]{}\mathbf{else}\;{}\<[16]%
\>[16]{}\Varid{m}\mathrel{\hstretch{0.7}{>\!\!>\!\!=}}\lambda \Varid{x}\to \emptyset\mathbin{\hstretch{0.7}{>\!\!>}}\Varid{return}\;\Varid{x}{}\<[E]%
\\
\>[B]{}\mathrel{=}{}\<[BE]%
\>[5]{}\mbox{\commentbegin  since \ensuremath{\Varid{return}\;()\mathbin{\hstretch{0.7}{>\!\!>}}\Varid{n}\mathrel{=}\Varid{n}} and \eqref{eq:bind-mzero-zero}  \commentend}{}\<[E]%
\\
\>[B]{}\hsindent{4}{}\<[4]%
\>[4]{}\mathbf{if}\;\Varid{p}\;\mathbf{then}\;\Varid{m}\;\mathbf{else}\;\Varid{m}\mathbin{\hstretch{0.7}{>\!\!>}}\emptyset{}\<[E]%
\\
\>[B]{}\mathrel{=}{}\<[BE]%
\>[5]{}\mbox{\commentbegin  assumption: \ensuremath{\Varid{m}\mathbin{\hstretch{0.7}{>\!\!>}}\emptyset\mathrel{=}\emptyset}  \commentend}{}\<[E]%
\\
\>[B]{}\hsindent{4}{}\<[4]%
\>[4]{}\mathbf{if}\;\Varid{p}\;\mathbf{then}\;\Varid{m}\;\mathbf{else}\;\emptyset{}\<[E]%
\\
\>[B]{}\mathrel{=}{}\<[BE]%
\>[5]{}\mbox{\commentbegin  since \ensuremath{\Varid{return}\;()\mathbin{\hstretch{0.7}{>\!\!>}}\Varid{n}\mathrel{=}\Varid{n}} and \eqref{eq:bind-mzero-zero}  \commentend}{}\<[E]%
\\
\>[B]{}\hsindent{4}{}\<[4]%
\>[4]{}\mathbf{if}\;\Varid{p}\;\mathbf{then}\;\Varid{return}\;()\mathbin{\hstretch{0.7}{>\!\!>}}\Varid{m}{}\<[E]%
\\
\>[4]{}\hsindent{5}{}\<[9]%
\>[9]{}\mathbf{else}\;\emptyset\mathbin{\hstretch{0.7}{>\!\!>}}\Varid{m}{}\<[E]%
\\
\>[B]{}\mathrel{=}{}\<[BE]%
\>[5]{}\mbox{\commentbegin  by \eqref{eq:if-distr}, with \ensuremath{\Varid{f}\mathrel{=}(\mathbin{\hstretch{0.7}{>\!\!>}}\Varid{m})}  \commentend}{}\<[E]%
\\
\>[B]{}\hsindent{4}{}\<[4]%
\>[4]{}(\mathbf{if}\;\Varid{p}\;\mathbf{then}\;\Varid{return}\;()\;\mathbf{else}\;\emptyset)\mathbin{\hstretch{0.7}{>\!\!>}}\Varid{m}{}\<[E]%
\\
\>[B]{}\mathrel{=}{}\<[BE]%
\>[5]{}\mbox{\commentbegin  definition of \ensuremath{\Varid{guard}}  \commentend}{}\<[E]%
\\
\>[B]{}\hsindent{4}{}\<[4]%
\>[4]{}\Varid{guard}\;\Varid{p}\mathbin{\hstretch{0.7}{>\!\!>}}\Varid{m}~~.{}\<[E]%
\ColumnHook
\end{hscode}\resethooks
\end{proof}

\end{document}